\documentclass[12pt, draftclsnofoot, journal, onecolumn, letterpaper,
]{IEEEtran}

\usepackage{stfloats}
\usepackage{amsfonts}
\usepackage{amssymb}
\usepackage{amsthm}
\usepackage{cite}
\usepackage[cmex10]{amsmath}
\usepackage{float}
\usepackage{color}
\usepackage{algorithm}
\usepackage{algorithmic}
\usepackage{multirow}
\usepackage[normalem]{ulem}
\usepackage{tabularx}
\usepackage{subfigure}
\usepackage{fancybox,dashbox}
\usepackage{bbm}

\newtheorem{theorem}{Theorem}
\newtheorem{lemma}{Lemma}

\usepackage{graphicx}

\begin{document}
	
	\title{A Learning-Based Two-Stage Spectrum Sharing Strategy with Multiple Primary Transmit Power Levels}
	\author{\IEEEauthorblockN{Rui Zhang, Peng Cheng, Zhuo Chen, Yonghui Li, and Branka Vucetic}
	\thanks{R.~Zhang, P.~Cheng, Y.~Li, and B.~Vucetic are with the School of Electrical and Information Engineering, the University of Sydney, Australia, (e-mail: rui.zhang1@sydney.edu.au; peng.cheng@sydney.edu.au; yonghui.li@sydney.edu.au; branka.vucetic@sydney.edu.au). Z.~Chen is with CSIRO DATA61, Australia (e-mail: zhuo.chen@ieee.org).}}
	
\markboth{Accepted by IEEE Transactions on Signal Processing}%
{Shell \MakeLowercase{\textit{et al.}}: Bare Demo of IEEEtran.cls for Computer Society Journals}
	
	\maketitle
	\begin{abstract}
Multi-parameter cognition in a cognitive radio network (CRN) provides a more thorough understanding of the radio environments, and could potentially lead to far more intelligent and efficient spectrum usage for a secondary user. In this paper, we investigate the multi-parameter cognition problem for a CRN where the primary transmitter (PT) radiates multiple transmit power levels, and propose a learning-based two-stage spectrum sharing strategy. We first propose a data-driven/machine learning based multi-level spectrum sensing scheme, including the spectrum learning (Stage I) and prediction (the first part in Stage II). This fully blind sensing scheme does not require any prior knowledge of the PT power characteristics. Then, based on a novel normalized power level alignment metric, we propose two prediction-transmission structures, namely periodic and non-periodic, for spectrum access (the second part in Stage II), which enable the secondary transmitter (ST) to closely follow the PT power level variation. The periodic structure features a fixed prediction interval, while the non-periodic one dynamically determines the interval with a proposed reinforcement learning algorithm to further improve the alignment metric. Finally, we extend the prediction-transmission structure to an online scenario, where the number of PT power levels might change as a consequence of PT adapting to the environment fluctuation or quality of service variation. The simulation results demonstrate the effectiveness of the proposed strategy in various scenarios.
	\end{abstract}
	
	\begin{IEEEkeywords}
		Cognitive radio, multiple primary transmit power levels, machine learning.
	\end{IEEEkeywords}
	
	\section{Introduction}
	\subsection{Cognitive Radio}
	The emerging new wireless technologies, such as 5G cellular networks and machine-to-machine enabled industrial Internet of Things, are fueling an ever-increasing demand for access to the radio frequency spectrum. Cognitive radio (CR), an intelligent wireless technology able to recognize the surrounding radio environments \cite{1391031}, creates a potential communication paradigm to achieve more efficient and flexible spectrum usage. A secondary user (SU) with CR capability monitors the spectrum utilization of a primary user (PU) and determines its access to such spectrum accordingly. Two fundamental challenges arise in the process: how to explore possible spectrum opportunities (spectrum sensing) and how to exploit such opportunities efficiently (spectrum access).
	
	Spectrum sensing measures and percepts the surrounding radio spectrum state based on various signal processing methods, including matched filter detection \cite{1399240,5175440}, cyclostationary detection \cite{4413137,4567443}, and energy detection \cite{4796930,5703204,6178840,5351690,6679037}. Matched filter can achieve the optimum performance, but the SU requires perfect knowledge of the PU signaling features \textit{a priori}. Cyclostationary detection takes advantage of the signal cyclostationary feature to distinguish signals from the stationary noise. In contrast, energy detection carries out the hypothesis test to determine the PU spectrum state, based on the energy of the collected PU signals. It features a low computational complexity, and is widely adopted in the literature. Upon obtaining the radio spectrum state, spectrum access dynamically adjusts the resources available to the SU, including frequency band, transmission time and transmission power, and accesses the licensed spectrum by taking into account the interference to the PU. The SU can access the licensed spectrum either when the PU is idle (opportunistic access) \cite{4494711}, or concurrently with the PU following a power control strategy to constrain the interference to the PU (spectrum sharing) \cite{4100173,4205091}. It is clear that appropriately designed spectrum sharing achieves higher throughput for the SU compared with the opportunistic spectrum access.
	
	It is worth noting that many contemporary wireless standards, such as IEEE 802.11 \cite{icslms1997wireless}, GSM \cite{mouly1992gsm}, and LTE \cite{remy2014lte}, have specified multiple transmit power levels to dynamically adapt to the fast changing environment and varying quality of service (QoS). The majority of CR studies in the literature did not take this into account, and the SU usually adopts a binary approach in reporting the radio spectrum state as idle or busy. In fact, for this multiple power level scenario, multi-parameter cognition is required, and the binary approach may not represent the most efficient spectrum utilization for the SU. The question then arises: how to exploit the variation in the PU power levels to design an intelligent spectrum sharing strategy? Naturally the answer to this question shall consist of spectrum sensing and spectrum access, which will be elaborated on in the next two subsections. 

	\subsection{Multi-level Spectrum Sensing}
	\label{Section_I_B}
	\begin{figure}[!t]
        \centering   
        \includegraphics[width=0.7\textwidth]{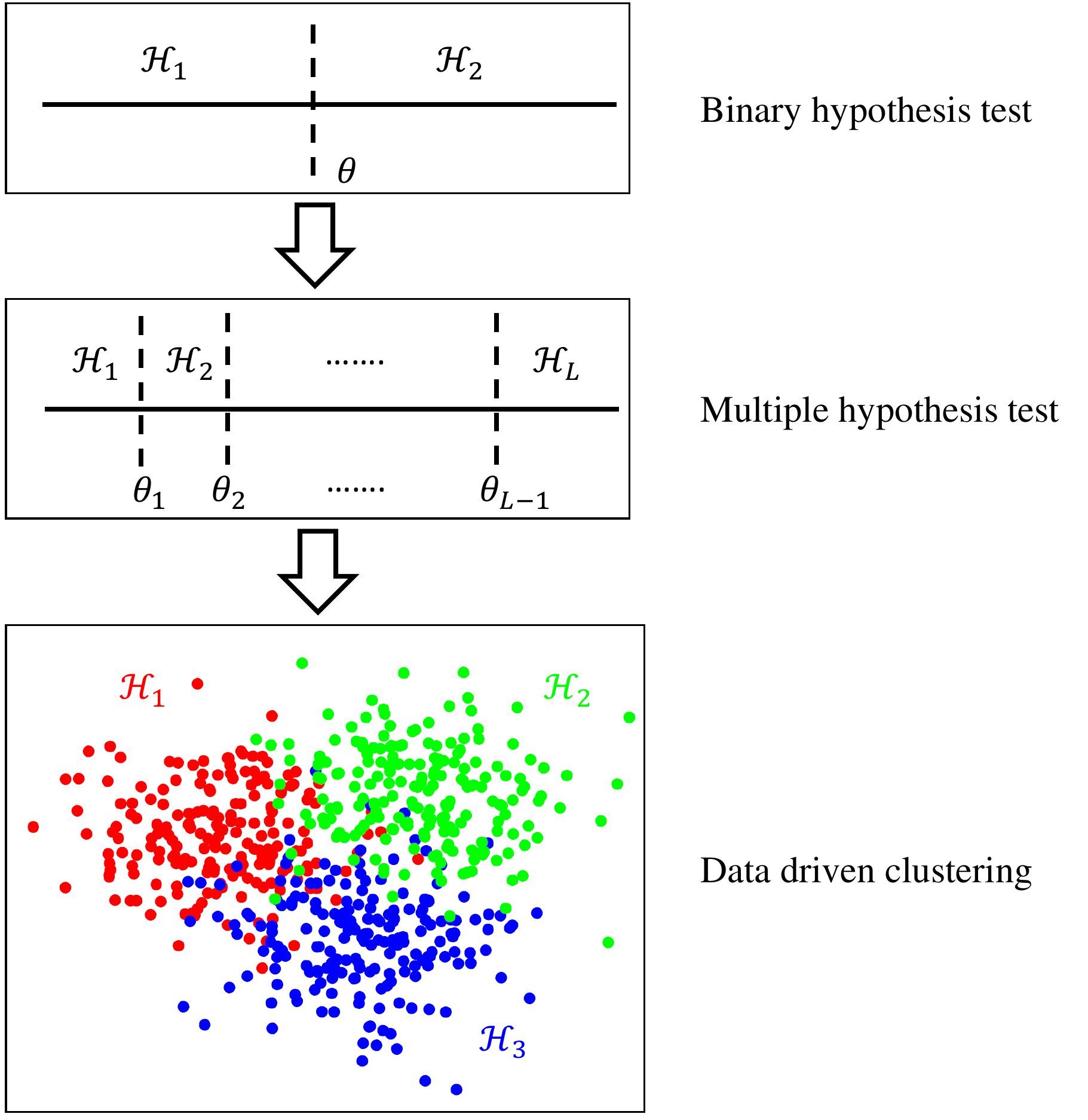}
        \caption{Spectrum learning: from hypothesis test to blind data driven.}        
        \label{binary_mul_data}
    \end{figure} 
    
	The priority in the design of an agile spectrum sensing method should aim to accurately map the sensing samples received by the secondary transmitter (ST) to the corresponding primary transmitter (PT) power levels. In a conventional \textit{binary} case (PT is ON/OFF), there are two kinds of errors, and the goal of the spectrum sensing is to determine a detection threshold $\theta$ (the upper part of Fig. \ref{binary_mul_data}). For example, given the target probability of false alarm and the noise power, $\theta$ can be simply determined by the Neyman-Pearson criterion\cite{9004667}. By contrast, the goal of the \textit{multi-level} case is to jointly determine multiple thresholds $\{\theta_i\}_{i=1}^{L-1}$ to separate $L$ different power levels through a multiple hypothesis test (the middle part of Fig. \ref{binary_mul_data}), which is far more complicated than the binary one. In essence, there are $L(L-1)$ kinds of errors, which are intertwined to exacerbate the complexity in threshold calculation. In \cite{7065293} and \cite{7911241}, the authors proposed an energy detection based multiple hypothesis test to derive the decision thresholds for the multiple power level identification. The results were extended to the scenarios with noise uncertainty \cite{7482689} and non-Gaussian transmission signals \cite{8166804}. However, the calculation of the thresholds in \cite{7065293,7911241,7482689,8166804} requires a large amount of prior knowledge at the ST, including the noise power and the PT transmit power mode (i.e., the number and exact value of different transmit powers, and the prior probability of each hypothesis). In practice, these parameters are unlikely to be available to the ST \textit{a priori}.
	
    In this paper, we aim to break the limits of the existing work, and achieve multi-level spectrum sensing with no or minimal prior information. We deviate from the above classical signal processing approaches, and machine learning arises as the tool of our choice for knowledge discovery to mine and extract the latent patterns reflecting the PT power level variation in the PT traffic data flow. On this basis, we propose a data-driven/machine learning based multi-level spectrum sensing scheme. It is fully blind in the sense that the ST does not require any prior knowledge of the noise power and the PT transmit power mode. Specifically, the proposed spectrum sensing scheme spans across two stages as shown in Fig. \ref{Power Model}. In Stage I (spectrum learning, a.k.a. the training phase in machine learning), the ST collects a multitude of signals that experience multiple PT transmit power level variation, and uses the Gaussian mixture models (GMMs) to capture the multi-level power characteristics inherent in the signals. Then, we introduce a Bayesian nonparametric method, referred to as conditionally conjugate Dirichlet process GMM (DPGMM), to automatically cluster the signals with the same PT transmit power level (the lower part of Fig. \ref{binary_mul_data}) and infer the model parameters (GMM parameters and PT power level duration distribution parameters). With the model parameters inferred in Stage I, the prediction part in Stage II (see Fig.~\ref{Power Model}) can easily identify the current PT power level through collecting PT signal samples. In this case, Stage I together with the prediction part in Stage II achieve fully blind multi-level spectrum sensing. Note that the second part in Stage II (ST transmission) will be detailed~next.
    
    \subsection{Multi-level Spectrum Access}
    With the big picture of multi-level PT radio environments learned in Section \ref{Section_I_B}, we next aim to establish a multi-level spectrum access strategy, which is the ultimate goal of spectrum sensing. To the best of our knowledge, this represents the first effort in such direction. For the SU, there is a fundamental tradeoff between two conflicting goals, namely, maximization of its own throughput and minimization of its interference to the PU. As typically there is no cooperation between the PU and SU, it is extremely difficult to optimize this tradeoff in practice. To provide a pragmatic solution to this dilemma, we first propose a new metric, referred to as the normalized power level alignment (NPLA), and it is defined as the time proportion that the ST matches its transmit power level to that of the PT.

    To optimize the NPLA performance, we propose two prediction-transmission structures (periodic and non-periodic) in Stage II for spectrum access which enables the ST to closely follow the PT power level variation. As discussed before, the prediction part can identify the current PT power level. On this basis, the transmission part adjusts the ST power according to the required signal-to-interference-plus-noise ratio (SINR) at the primary receiver (PR). Specifically, the periodic structure features a fixed prediction interval, and is straightforward in implementation. By contrast, the non-periodic structure dynamically determines the interval, which can be formulated as a partially observable decision problem. This motivates us to develop a new algorithm based on reinforcement learning \cite{reinlearning}, exploiting the PT power level duration distribution. This structure further improves the NPLA performance. Finally, we extend the prediction-transmission structure to an online scenario, where the number of PT power levels might change as a consequence of PT adapting to the environment fluctuation or QoS variation.
	
    \subsection{Contribution}
	In a nutshell, we propose a learning-based two-stage spectrum sharing strategy for a CR network, enabling fully blind spectrum sensing when the PT power varies with time in multiple levels, and designs an adaptive spectrum access strategy for the NPLA optimization. The main contributions of this paper can be summarized as follows.
	\begin{itemize}
	    \item We propose a novel data-driven/machine learning based multi-level blind spectrum sensing. The conditionally conjugate DPGMM with Bayesian inference is introduced to automatically cluster the signals and infer the model parameters, which is key to predict the PT power levels. 
        \item We propose a new metric, NPLA, to strike an excellent tradeoff between the secondary network throughput and the interference to the primary network.
	    \item To optimize the NPLA performance, we propose a prediction-transmission structure for spectrum access which enables the ST to closely follow the PT power level variation. Furthermore, the ST prediction interval is dynamically adjusted to achieve better performance.
	    \item The spectrum access method is extended to the online scenario to accommodate a more realistic situation, where the number of PT power levels might change after the inital spectrum learning.
	\end{itemize}

	\subsection{Related Work}
	Machine learning technology has recently played an important role in improving spectrum sensing. The work in \cite{8403655} presented an adversarial machine learning approach to launch jamming attacks on CR and introduces a defense strategy. Several supervised and unsupervised machine learning algorithms for cooperative spectrum sensing (CSS) were investigated in \cite{6635250}.	In \cite{6096448}, the combination of infinite GMM and CSS was proposed to detect the primary user emulation attacks. In \cite{1705.08164}, a convolutional neural network-based CSS scheme was developed to detect multiple bands simultaneously. A mobile CSS framework was proposed in \cite{8466022} for  large-scale heterogeneous cognitive networks.

	There are many efforts in sensing policy design for real-time decisions on which channel(s) to sense (dynamic multichannel selection). By contrast, our paper considers the single user and single channel case, and focuses on the design of multi-level spectrum sensing to differentiate different PT power levels. On this basis, we also consider the policy design to dynamically adjust the sensing intervals to improve the NPLA performance.
    
    The dynamic multichannel selection can be modelled as a partially observable Markov decision process (POMDP)\cite{4155374}. The partial observation in \cite{6200864,6362216,8437583,8303773,8532121,8359094,Naparstek2017DeepMR} originates from each SU being unable to scan all the channels at any one time due to energy and hardware constraints. Therefore, a sensing policy needs to be developed to balance between utilizing a spectrum opportunity for immediate access and collecting spectrum occupancy statistics to track spectrum opportunity for future exploitation. As finding an optimal channel sensing policy in general is computationally prohibitive with the increased number of channels, several efforts endeavor to find the optimal/near-optimal policy with low computational cost. In \cite{6200864,6362216,8437583}, the dynamic multichannel access problem is modelled as a restless multi-armed bandit problem. The time horizon is divided by interleaving exploration and exploitation epochs with growing lengths, and the optimal policy can be translated into determining the length and allocation of each epoch. Recently, deep reinforcement learning (DRL) based channel selection \cite{8303773,Naparstek2017DeepMR,8532121,8359094} has attracted great attention, and it aims to handle the correlated channels with unknown channel dynamics. The essence of DRL is to provide a good approximation of the objective value (Q-value), facilitating the handling of the large state and action spaces.

    It is worth noting that the access policy design (sensing interval) in our paper is also formulated as a POMDP, but the nature of our formulation is fundamentally different from that in \cite{6200864,6362216,8437583,8303773,Naparstek2017DeepMR,8532121,8359094}. The partial observation in our work comes from imperfect multi-level sensing results and access feedback. To tackle this challenging POMDP, we reduce the infinite time horizon to a finite one, leading to a computationally tractable solution. Most importantly, we mathematically prove that such practice does not sacrifice the optimality in the utility.

    \subsection{Organization and Notation}
	The rest of the paper is organized as follows. In Section \ref{System Model}, we discuss the system model for our proposed two-stage spectrum sharing strategy. In Section \ref{Bayesian Non}, we introduce a Bayesian nonparametric method and its inference for the model parameters. The prediction-transmission structure with an online extension, which are adaptive to the PT power level variation, is presented in Section \ref{Section Spectrum}. Simulation results and discussion are presented in Section \ref{Section Numerial Results} followed by conclusions in Section \ref{Conclustion}.
	
	$\mathcal{N}(\mu,S^{-1})$ denotes the Gaussian distribution with mean $\mu$ and precision $S$, $\mathcal{CN}(\mu,S^{-1})$ denotes the complex Gaussian distribution with mean $\mu$ and precision $S$, and $\mathcal{G}(a,b)$ denotes the Gamma distribution with shape parameter $a$ and scale parameter $b$. $\Gamma(\cdot)$ denotes the Gamma function and $\Gamma(\cdot,\cdot)$ denotes the incomplete Gamma function. $\left \lfloor \cdot \right \rfloor$ is the floor function. For convenience, we list most important symbols in Table~\ref{Table of Symbols}.
	
	\begin{table*}
		\caption{List of Symbols}
		\label{Table of Symbols}
		\begin{center}
			\begin{tabular}{c|c}
				\hline
				Symbol & Definition \\ 
				\hline
				\hline
				$l$, $L$ & The index and total number of the actual PT transmit power levels. \\
				\hline
				$k$, $K$ & The index and total number of the PT transmit power level estimated by the ST. \\
				\hline
				$m$, $M$ & The index of the PT hypothesis and the total number of the PT hypotheses in Stage I. \\
				\hline
				$n$, $N$ & The index of the ST action and the total number of sensing slots in Stage I. \\
				\hline
				$P_{PT,l}$ & The PT transmission power value on the $l$-th level. \\
				\hline
				$\mathcal{H}_l$ & The hypothesis that the PT transmits with $P_{PT,l}$.\\
				\hline 
				$\mathcal{\hat{H}}_{k}$ & The hypothesis determined by the ST that the PT transmits with $P_{PT,k}$. \\
				\hline
				$X_n$ & The test statistic in the $n$-th single sensing slot. \\
				\hline
				$N^{(s)}$ & The total number of samples collected by the ST in a single sensing slot. \\
				\hline
				$T_{ss},T_{st},T_{po}$ & The duration of a ST sensing slot, a ST transmission block, and a PT hypothesis. \\
				\hline
				$\tau_s$, $\tau_p$ & The discretized time of a ST transmission block and a PT hypothesis.\\
				\hline
				$\alpha$, $G_0$ & The concentration parameter and the base probability distribution of the Dirichlet process. \\
				\hline
				$z_n$ & The latent variable indicating which component that $X_n$ is associated with. \\
				\hline
				$N_k$ & The total number of observations assigned to the $k$-th component. \\
				\hline
				$\mu_k,S_k,\pi_k$ & The value, precision, and mixing proportion of the $k$-th component in the GMM. \\
				\hline
				$\lambda$, $R$, $W$, $\beta$ & The hyperparameters in the conditionally conjugate DPGMM. \\
				\hline
				$H_{kj}$ & The probability that the PT is operating under $\mathcal{H}_k$ while the detection by the ST is in favor of $\mathcal{H}_j$. \\
				\hline
				$C_{kj}$ & The probability that the PT transfers from the $k$-th transmit power level to the $j$-th level. \\
				\hline
				$p_{\tau}^k$ & The probability that the PT keeps operating with the $k$-th transmit power level at time $\tau$. \\
				\hline
				$a_{\tau}$ & The action that the ST will take at time $\tau$. \\
				\hline
				$\mathcal{T}_k$ & The longest time that the SU should transmit when operating in the $k$-th transmit power level. \\
				\hline
				$P_c$ & The probability of correct PT power level prediction in Stage II. \\
				\hline
				$U(\tau)$ & The NPLA performance from time 0 to $\tau$. \\
				\hline
			\end{tabular}
		\end{center}
	\end{table*}
	
	\section{System Model}
	\label{System Model}
	We consider a spectrum sharing CR network in Fig. \ref{system model}, with a primary network consisting of a PT and a PR, a secondary network consisting of a ST and a secondary receiver (SR), and a central site (broadcast tower). Transmission happens simultaneously in both networks sharing the same frequency band. Different from most CR networks considered in the literature, the PT operates with multiple (instead of binary: ON/OFF) power levels. The ST attempts to learn the model parameters, which will be defined later, and then optimize the spectrum access accordingly. The ultimate goal is to optimize the NPLA performance.
	
	To achieve this target, we propose a novel two-stage spectrum sharing strategy, as illustrated in Fig. \ref{Power Model}. Let $P_{PT,l}$, $l = 1,\cdots,L,$ be the transmit powers of the PT, where $P_{PT,1} < P_{PT,2} < \cdots < P_{PT,L-1}$ and $P_{PT,L} = 0$ indicates an idle PT. For convenience, hypothesis $\mathcal{H}_l$ indicates that the PT transmits with power $P_{PT,l}$. It is assumed that $P_{PT,l}$ undergoes a slow change, as shown in the figure. We define the time duration of each hypothesis as a random variable $T_{po}$, which is usually much larger than that of the ST sensing slot $T_{ss}$ and the ST transmission block $T_{st}$. In this paper, we consider a time discretized model, where $T_{ss}$ is the minimum time unit. We define $\tau_s = T_{st}/T_{ss}$ as the fixed discretized time duration of the ST transmission block, and $\tau_p = T_{po}/T_{ss}$ as the varied discretized PT power level duration. As discussed before, the prior knowledge on the PT transmit power mode, defined as the number of transmit power levels $L$, the exact values $P_{PT,l}$, and the prior probability of each hypothesis $\Pr \{\mathcal{H}_l \}$, is unknown to the ST, which is fundamentally different from the assumptions in \cite{7065293,7911241,7482689,8166804}. In the sequel, we describe in detail the operations of these two stages in Fig. \ref{Power Model}. 
	
	\begin{figure}
		\centering
		\includegraphics[width=0.6\textwidth]{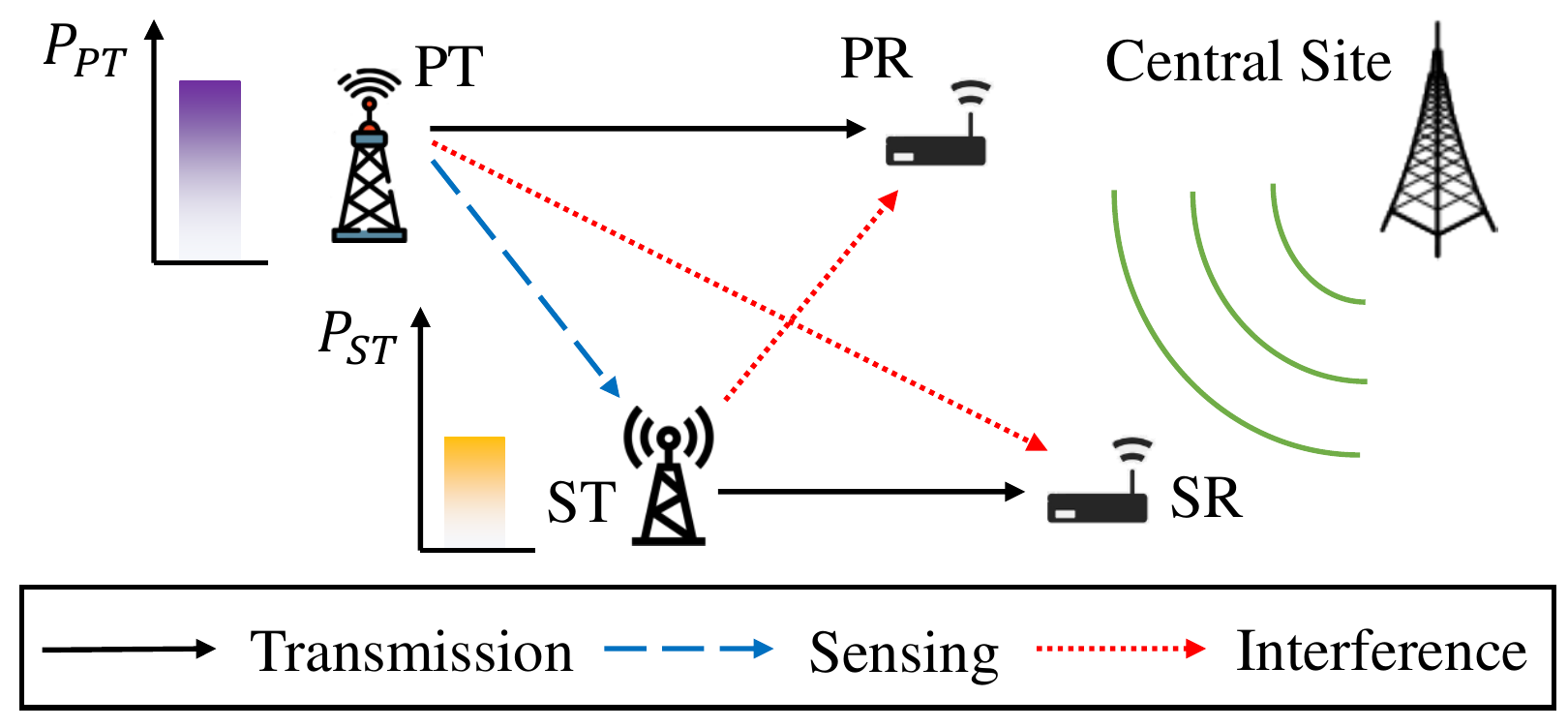}
		\caption{System model for a spectrum sharing CR network.}
		\label{system model}
	\end{figure}	
	
	\begin{figure*}
		\centering
		\includegraphics[width=0.95\textwidth]{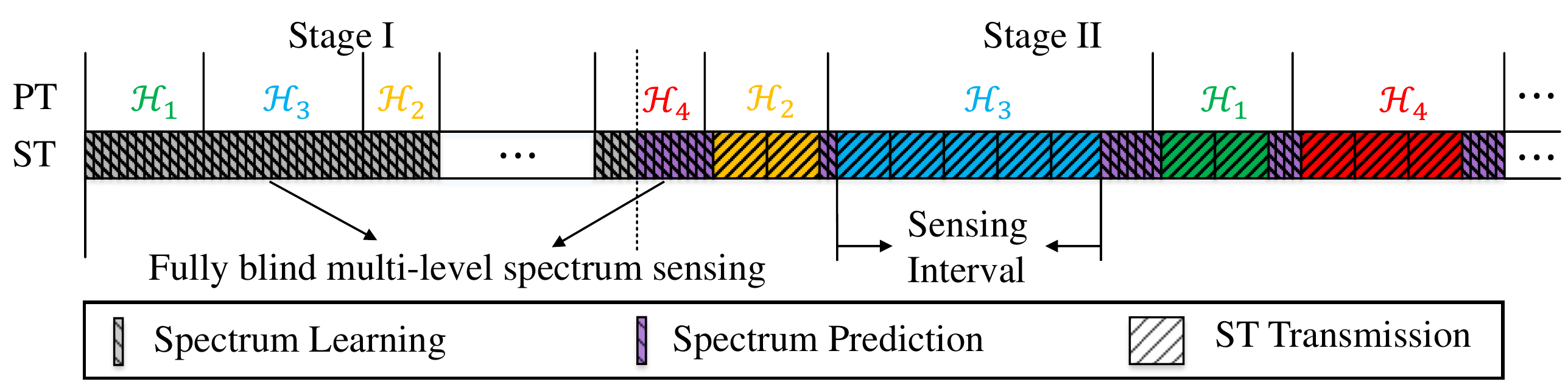}
		\caption{The proposed two-stage spectrum sharing strategy. The sensing slots of the ST in both stages have the same time duration $T_{ss}$. The sensing slots in Stage I are used for learning, while that in Stage II are for prediction.}
		\label{Power Model}
	\end{figure*}
	
	\subsection{Stage I}
	In this stage, the ST samples the received PT signals at a sampling frequency $f_s$ and collects $N^{(s)}$ samples in each sensing slot with duration $T_{ss}$ (for notation simplicity, we assume that $N^{(s)} = T_{ss} f_s$ is an integer). The ST observes $N$ sensing slots in Stage I and collects a total of $N N^{(s)}$ samples. It is assumed that the learning period is reasonably large so that it covers all $L$ hypotheses\footnote{There is a none-zero probability that some transmit power levels do not happen and are not observed by the ST during Stage I, even if the learning period is relatively large. In this case, these missed hypotheses can be viewed as small probability events. Consequently, they will have negligible impact on the performance of the proposed spectrum access method, and can be ignored.}. Thus, the $i$-th sample in the $n$-th sensing slot under hypothesis $\mathcal{H}_l$ can be given by
	\begin{equation}
	\label{R_i}
	R_{n}[i] = \sqrt{P_{PT,l}} s_n[i]+ u_n[i], {\ \ \ } \mathcal{H}_l, i = 1,\cdots,N^{(s)}; n= 1,\cdots, N,
	\end{equation}
    	where $\sqrt{P_{PT,l}} s_n[i]$ is the received primary signal in the $n$-th sensing slot with average power $P_{PT,l}$, and $u_n[i] \sim \mathcal{CN}(0,\sigma_u^2) $ is the additive white Gaussian noise. Following \cite{4489760}, we assume that $s_n[i]$ is an independent and identically distributed (i.i.d.) random variable with mean $0$ and variance $1$. Following \cite{4489760} and without loss of generality, we assume that $s_n[i]$ is complex PSK modulated signal\footnote{For other modulation schemes, the test statistic $X_n$ still follows a Gaussian distribution \cite{4489760}. Therefore, our proposed method is still valid for other modulation and/or adaptive modulation schemes.}.

	 The test statistic in the $n$-th sensing slot can be calculated as
	\begin{equation}
	\label{X_i}
	X_n = \dfrac{1}{N^{(s)}} \sum_{i=1}^{N^{(s)}} \left| R_{n}[i] \right| ^2, \mbox{\ \ \ } \mathcal{H}_l.
	\end{equation}
	When $N^{(s)}$ is large, according to the central limit theorem, the distribution of $X_n$ under hypothesis $\mathcal{H}_l$ can be approximated by a Gaussian one, and we have
	\begin{equation}
	\label{X_i_dis}
	X_n \sim 	\mathcal{N} \left( (\gamma_{st}^l + 1)\sigma_u^2,\dfrac{1}{N^{(s)}} (2\gamma_{st}^l + 1) \sigma_u^4 \right), \mbox{\ \ \ } \mathcal{H}_l,
	\end{equation}
	where $\gamma_{st}^l = P_{PT,l}/\sigma_u^2$ is the received signal-to-noise ratio (SNR) at the ST. Considering all the hypotheses, we establish that $X_n$ follows a mixed Gaussian distribution \cite{rasmussen2000infinite}
	\begin{equation}
	\label{mixG}
	X_n \sim \sum_{l = 1}^{L} \pi_l \mathcal{N}(\mu_l, S_l^{-1}),
	\end{equation}
	where $0 \leqslant \pi_l \leqslant 1$ is the mixing coefficients with $\sum_{l = 1}^{L} \pi_l = 1$. Each Gaussian density $\mathcal{N}(\mu_l, S_l^{-1})$ is a component of the mixture with mean value $\mu_l = (\gamma_{st}^l + 1)\sigma_u^2$ and precision $S_l = \left( (2\gamma_{st}^l + 1) \sigma_u^4 / N^{(s)} \right) ^{-1}$. 
	
	Given the observation set $\mathbf{X} = \{ X_1,\cdots,X_N \}$, the proposed Bayesian nonparametric method aims to infer the GMM parameter set $\{\boldsymbol{\theta}, \boldsymbol{\pi}, L \}$, where $\boldsymbol{\theta} = \{\theta_1,\cdots,\theta_L \}$ with $\theta_l = \{\mu_l, S_l \}$ and $\boldsymbol{\pi} = \{\pi_1,\cdots,\pi_L\}$. In other words, our method automatically clusters the signals with the same PT transmit power levels. In summary, Stage I establishes a big picture of the PT activities at the cost of an one-off overhead. After learning, the ST allocates the same number of power levels as that in the PT, with an initial ST power value $P_{ST,k}$ for each level $k$.

	\subsection{Stage II}
	In this stage, we propose two prediction-transmission structures (periodic and non-periodic) adapting to the PT power level variation for spectrum access. The main features of the structures can be summarized as follows.
	\begin{itemize}
		\item As shown in Fig. \ref{Power Model}, Stage II consists of two parts: prediction and transmission actions. In the prediction action (the sensing slots with purple color), the ST can easily identify the current PT power level $l$, which is jointly determined by the test statistic $X_n,n > N$, in \eqref{X_i} and the inferred GMM parameter set $\{\boldsymbol{\theta}, \boldsymbol{\pi}, L \}$. In the transmission action, the ST allocates its transmit power level $k$ to match the latest identified PT power level $l$ ($k = l$). Here, the corresponding $P_{ST,k}$ can be determined as follows. Assume that the required SINR for the PR is $\Gamma_0$ and the current received SINR is  $\Gamma_{PR}$. A nearby monitoring station (see Fig. \ref{system model}) of the PR transmits $\Gamma_{PR}$ to the central site, likely through optical fiber or microwave, which then broadcasts this information on a dedicated frequency. We assume that the ST is able to decode broadcast signals from the central site and communicate with the SR on different frequency bands. Through the broadcast nature, the ST obtains $\Gamma_{PR}$. If $\Gamma_{PR}<\Gamma_0$, the ST should reduce the transmit power and vice versa. In other words, $P_{ST,k}$ for each level $k$ will gradually approach a desired power value. This guarantees that the PR is well protected, while the secondary network obtains the highest possible throughput\footnote{The similar idea of using broadcasting mechanism was also suggested by Federal Communications Commission \cite[\text{p. 6}]{fcc0327} and adopted in 4G LTE systems in the form of inter-cell interference overload indicator \cite{letmag}.}. On this basis, it is clear that the alignment between ST and PT power levels can optimize the trade-off between the interference to primary network and throughput of the secondary network\footnote{Note that the channel state information for PT-ST and ST-PR is not required in our approach. In addition, the use of the received SINR has already included the impact of the channel fading.}. In the case that the ST mismatches the PT power level variation, either the PR is interfered below the required SINR or unnecessarily lower secondary network throughput is obtained.
		\item For the periodic structure, the prediction intervals are fixed. By contrast, in the non-periodic structure, the intervals are dynamically determined to enhance the NPLA performance. As shown in Fig. \ref{Power Model}, zero intervals are used to track the PT power level variation, while long intervals are selected to avoid unnecessary prediction. The  non-periodic structure will be elaborated on in Section \ref{SectionIVC}.
	\end{itemize}
	
	\section{Spectrum Learning Based on Bayesian Inference}
	\label{Bayesian Non}
	In this section, we focus on Stage I, and introduce a Bayesian nonparametric method to infer the GMM parameter set $\{\boldsymbol{\theta}, \boldsymbol{\pi}, L\}$ based on the observation set $\mathbf{X}$. As $L$ is unknown \textit{a priori}, the traditional methods, such as the K-mean and expectation maximization, are inapplicable. This motivates us to resort to Dirichlet process Gaussian mixture model (DPGMM)\cite{Gorur2010}, which takes into account the Gaussian mixture property and is able to identify the unknown number of Gaussian components. For specific Bayesian inference, we choose the Markov chain Monte-Carlo based Gibbs sampling method \cite{doi:10.1080/10618600.2000.10474879} considering its simplicity.
	
	In the following, we first review the preliminary knowledge on the Dirichlet process mixture model. On this basis, we introduce the DPGMM considering the specific distribution of the observation set $\mathbf{X}$. Furthermore, we modify the DPGMM to the conditionally conjugate case to simplify the inference process. Finally, we carry out Bayesian inference with Gibbs sampling method to infer $\{\boldsymbol{\theta}, \boldsymbol{\pi}, L \}$.
	
	\subsection{Dirichlet Process Mixture Model}
	The Dirichlet distribution is an extension of the Beta distribution for multivariate cases. It represents the probability of $K$ events given that the $k$-th event $x_k$ ($k = 1,\cdots,K$) has been observed $\alpha_k - 1$ times. The probability density function can be expressed as
	\begin{equation}
	\mathrm{Dir}(\alpha_1,\cdots,\alpha_K) = \dfrac{\Gamma\left( \sum_{k = 1}^K \alpha_k \right) }{\prod_{k = 1}^K \Gamma(\alpha_k)} \prod_{k = 1}^{K} \pi_k^{\alpha_k - 1},
	\end{equation}
	where $\pi_k$ is the probability of the $k$-th event $x_k$ with $\sum_{k = 1}^{K} \pi_k = 1$ and $\pi_k > 0$.
	
	In our application, the event $x_k$ represents the $k$-th possible PT transmit power level, which can not be observed explicitly. Instead, the explicit observation is the test statistic $X_n$. As $X_n$ is drawn from a distribution based on event $x_k$, we introduce the Dirichlet process (DP) to define the distribution of $X_n$. A random measure $G$ is said to be a Dirichlet process distributed with a base probability distribution $G_0$ and a concentration parameter $\alpha$, if we have
	\begin{equation}
	\label{DP_definition}
	\left( G(A_1),\cdots,G(A_i) \right) \sim \mathrm{Dir}(\alpha G_0(A_1),\cdots,\alpha G_0(A_i))
	\end{equation}
	for every finite measurable partition $\{A_1,\cdots,A_i\}$ of $\boldsymbol{\theta}$. It is written as $G \sim \mathrm{DP}(G_0,\alpha)$.
	
	Next we model the observation set $\mathbf{X}$ using the parameter $\boldsymbol{\theta}$ based on the DP mixture model. A DP mixture model is suitable for the clustering purposes, where the number of Gaussian components is not known \textit{a priori}. Here, $X_n$ can be regarded as an independent draw from the distribution $F(\theta_n)$, where each $\theta_n$ is an i.i.d. draw from a DP $G$. Mathematically, the DP mixture model can be expressed as \cite{teh2011dirichlet}
	\begin{equation}
	\label{DPMM}
	\begin{split}
	X_n | \theta_n & \sim F(\theta_n), \\
	\theta_n | G & \sim G, \\
	G | \{G_0,\alpha\} & \sim \mathrm{DP}(G_0,\alpha).
	\end{split}
	\end{equation}
	Note that the formulation \eqref{DPMM} represents the most general case. In our case, two different observations $X_n$ and $X_{n'}$ ($n \neq n'$) may follow the same distribution, and \eqref{DPMM} can not explicitly reveal such property. Therefore, following \cite{doi:10.1080/10618600.2000.10474879}, a latent variable $z_n$ is introduced to explicitly indicate which transmit power level that $X_n$ is associated with, and will be referred to as an indicator hereafter. Accordingly, an equivalent model can be obtained as
	\begin{equation}
	\label{DPGMM}
	\begin{split}
	X_n | \{\boldsymbol{z}, \boldsymbol{\phi} \} & \sim F(\phi_{z_n}), \\
	\phi_k | G_0 & \sim G_0, \\
	p ( z_n = k ) & = \pi_k, \\
	\boldsymbol{\pi} | \{ \alpha, K \} & \sim \mathrm{Dir}(\alpha/K,\cdots\alpha/K),
	\end{split}
	\end{equation}
	where $\boldsymbol{z} = \{ z_1,\cdots,z_N \}$ is the set of indicators, $\boldsymbol{\phi} = \{\phi_1,\cdots,\phi_{K}\}$ is the set of unique values in $\boldsymbol{\theta}$, and $\theta_n = \phi_{z_n}$. Hereafter, $K$ refers to the total number of Gaussian components, and each component consists of the observations that are determined by the ST as having the same transmit power level. Note that $\boldsymbol{\pi}$ is assumed to have a symmetric Dirichlet distribution, where all the concentration parameters are $\alpha/K$. This assumption is widely adopted when there is no prior knowledge of the mixing proportions \cite{doi:10.1080/10618600.2000.10474879}. Let $N_k$ denote the number of observations assigned to the $k$-th component, then $N_k$ follows a multinomial distribution
	\begin{equation}
	p(N_1,\cdots,N_K | \boldsymbol{\pi}) = \dfrac{N!}{N_1! \cdots N_K!} \prod_{k = 1}^{K} \pi_k^{N_k},
	\end{equation}
	where $\sum_{k=1}^{K} N_k = N$, and the distribution of the indicators~is
	\begin{equation}
	\label{p_z_pi}
	p(\boldsymbol{z} | \boldsymbol{\pi}) = \prod_{k = 1}^{K} \pi_k^{N_k}.
	\end{equation}
	We can integrate out the mixing proportions of the product of $p(\boldsymbol{\pi} | \alpha)$ in \eqref{DPGMM} and $p(\boldsymbol{z}|\boldsymbol{\pi})$ in \eqref{p_z_pi}, and the prior on $\boldsymbol{z}$ in terms of $\alpha$ is expressed as \cite{Gorur2010}
	\begin{equation}
	\begin{split}
	p(\boldsymbol{z}|\alpha) = \int p(\boldsymbol{z}|\boldsymbol{\pi}) p(\boldsymbol{\pi} | \alpha) \mathrm{d} \boldsymbol{\pi} 
	& = \int \dfrac{\Gamma (\alpha)}{\Gamma(\alpha/K)^K} \prod_{k = 1}^{K} \pi_k^{N_k + \alpha/K -1} \mathrm{d} \boldsymbol{\pi} \\
	& = \dfrac{\Gamma(\alpha)}{\Gamma(N + \alpha)} \prod_{k = 1}^{K} \dfrac{\Gamma(N_k + \alpha/K)}{\Gamma(\alpha/K)}.
	\end{split}
	\end{equation}
	As all the observations are exchangeable, if we assume that $\boldsymbol{z}_{-n} = \{z_1,\cdots,z_{n-1},z_{n+1},\cdots,z_N \}$ has been obtained, the conditional distribution for the individual indicator can be given by
	\begin{equation}
	\label{z_i}
	p(z_n = k | \boldsymbol{z}_{-n}, \alpha) = \dfrac{p(z_n = k, \boldsymbol{z}_{-n} | \alpha)}{p(\boldsymbol{z}_{-n} | \alpha)} = \dfrac{N_{-n,k} + \alpha/K}{N - 1 + \alpha},
	\end{equation}
	where $N_{-n,k}$ is the number of samples excluding $X_n$ in the $k$-th component. Similarly, the prior distribution of $\theta_n$ over $\boldsymbol{\theta}_{-n} = \{\theta_1,\cdots,\theta_{n-1},\theta_{n+1},\cdots,\theta_N \}$ can be written as
	\begin{equation}
	\label{theta_i}
	\theta_n | \boldsymbol{\theta}_{-n} \sim \frac{N_{-n,k}}{N-1+\alpha} \sum_{k=1}^{K} \delta(\phi_k) + \frac{\alpha}{N-1+\alpha} G_0.
	\end{equation}
	
	\subsection{Conditionally Conjugate Dirichlet Process Gaussian Mixture Model}
	Recall that the observation $X_n$ in \eqref{mixG} follows a mixed Gaussian distribution given by
	\begin{equation}
	\label{X_i_k}
	X_n | \{\boldsymbol{\mu}, \boldsymbol{S}, \boldsymbol{\pi} \} \sim \sum_{k= 1}^{K} \pi_k \mathcal{N}(\mu_k, S_k^{-1}),
	\end{equation}
	where $k = \{1,\dots,K\}$ replaces $l=\{1,\cdots,L\}$, as $k$ denotes the index of inferred hypotheses while $l$ is the index of the real hypothesis. Therefore, $X_n$ can be modeled as a DPGMM and expressed as
	\begin{equation}
	\label{DPGMM_X}
	\begin{split}
	X_n | \{ \boldsymbol{z}, \boldsymbol{\phi} \} & \sim \mathcal{N} (\mu_{z_n},S_{z_n}^{-1}),\\
	( \mu_k,S_k^{-1}) & \sim G_0, \\
	p ( z_n = k ) & = \pi_k, \\
	\boldsymbol{\pi} |\{ \alpha, K \} & \sim \mathrm{Dir}(\alpha/K,\cdots\alpha/K).
	\end{split}
	\end{equation} 
	In \eqref{DPGMM_X}, $G_0$ represents a prior guess of the distributions of $\mu_k$ and $S_k^{-1}$ in the DPGMM. Its choice is usually guided by mathematical convenience, and the conjugate form is widely adopted. In our case, the distribution of $G_0$ specifies the prior on the mixture Gaussian distributions parameters $\boldsymbol{\mu} = \{ \mu_1,\cdots,\mu_K \}$ and $\boldsymbol{S} = \{ S_1,\cdots,S_K \}$, and it can be expressed in a conjugate form \cite{rasmussen2000infinite}
	\begin{equation}
	\label{mu_S_j}
	\begin{split}
	\mu_k | \{ S_k, \xi, \rho \} & \sim \mathcal{N}(\xi,(\rho S_k)^{-1}), \\
	S_k | \{ \beta, W \} & \sim \mathcal{G}(\beta, W^{-1}),
	\end{split}
	\end{equation}
	where $\xi$, $\rho$, $\beta$ and $W$ are the hyperparameters for the DPGMM. It is clear that the prior distribution of $\mu_k$ is dependent on $S_k$. This undesirable property is inevitable due to the conjugacy requirement for $G_0$. To remove such dependency, we modify the original conjugate feature in the DPGMM and introduce the conditionally conjugate version of DPGMM to model $\mathbf{X}$. In a conditionally conjugate DPGMM, \eqref{mu_S_j} can be rewritten as \cite{rasmussen2000infinite}
	\begin{equation}
	\label{mu_S_j_2}
	\begin{split}
	\mu_k | \{ \lambda, R \} & \sim \mathcal{N}(\lambda,R^{-1}), \\
	S_k | \{ \beta, W \} & \sim \mathcal{G}(\beta, W^{-1}),
	\end{split}
	\end{equation}
	where $R$ is the hyperparameter. To complete the conditionally conjugate DPGMM and capture the features inherent in $\textbf{X}$, we need to give suitable values for the hyperparameters. However, the exact values are hard to know \textit{a priori}, and small changes on them will dramatically affect the model performance. To achieve the robustness for the model, we impose vague priors for the hyperparameters following \cite{Gorur2010},
	\begin{equation}
	\label{vague}
	\begin{split}
	\lambda & \sim \mathcal{N}(\mu_y, S_y^{-1}), \\
	R & \sim \mathcal{G}(1,S_y),\\
	W & \sim \mathcal{G}(1,S_y^{-1}),\\
	\beta^{-1} & \sim \mathcal{G} (1,1),
	\end{split}
	\end{equation}
	where the hyperpriors $\mu_y$ and $S_y$ refer to the empirical mean and precision of $\mathbf{X}$, respectively. In theory, the prior should not depend on the observations. However, as shown in \cite{Gorur2010}, the formulation for the priors in \eqref{vague} is equivalent to normalizing observations, and a wide range of parameters in the priors lead to similar inference results.
	
	The conditionally conjugate DPGMM in \eqref{vague} can be graphically represented in Fig. \ref{Hierarchical model}.
	
	\begin{figure}
		\centering
		\includegraphics[width=0.6\textwidth]{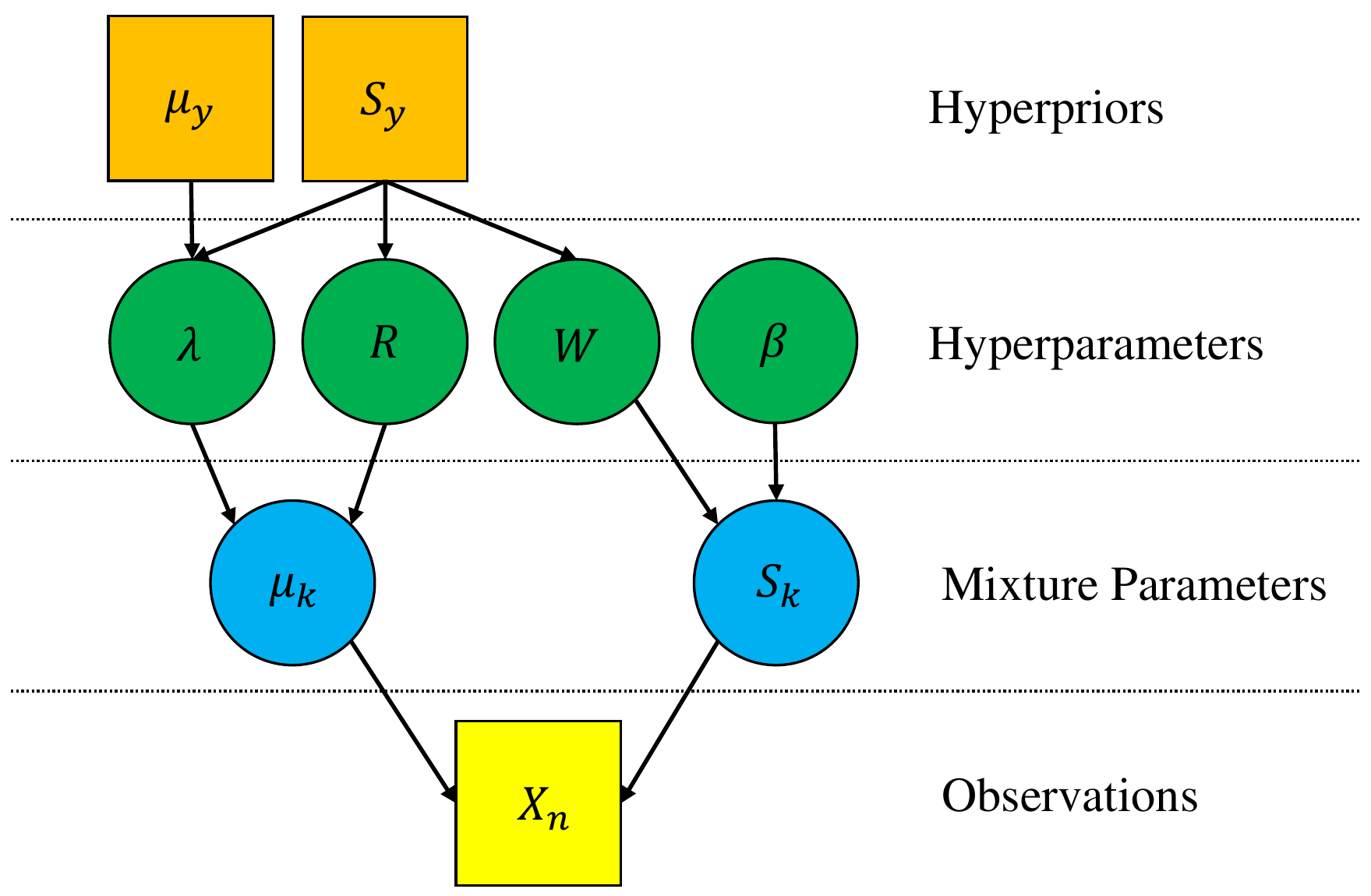}
		\caption{Graphical representation of the conditionally conjugate DPGMM.}
		\label{Hierarchical model}
	\end{figure}
	
	\subsection{Inference Using Gibbs Sampling}
	Given the conditionally conjugate DPGMM and observation $\textbf{X}$, now we use the Markov chain Monte-Carlo algorithm based on Gibbs sampling to infer $\{ \boldsymbol{\theta}, \boldsymbol{z}, L \}$. With the latent variable set $\boldsymbol{z}$, we can obtain the mixing proportion set $\boldsymbol{\pi}$. In the algorithm, we update the variables iteratively by sampling each variable from the posterior distribution conditioned on the others.
	
	Given the likelihood of $\mu_k$ and $S_k$ in \eqref{X_i_k} and their priors in \eqref{mu_S_j_2}, we can multiply the priors by the likelihood conditioned on $\boldsymbol{z}$ and obtain the conditional posterior distributions of $\mu_k$ and $S_k$, which can be given by
	\begin{equation}
    	\label{inferpara}
    	\begin{split}
    	\mu_k | \{ \boldsymbol{z}, S_k, \lambda, R \} & \sim
    	\mathcal{N}\left( \dfrac{ S_k \sum_{n = 1}^N \mathbbm{1}_{\mathbf{A}}\left(n\right) X_n  + \lambda R }{N_k S_k + \lambda}, \dfrac{1}{N_k S_k + \lambda} \right), \\
    	S_k | \{\boldsymbol{z}, \beta, W \} & \sim
    	\mathcal{G} \left( \beta + N_k, \left[ \dfrac{1}{\beta + N_k} \left(  W \beta + \sum_{n = 1}^N \mathbbm{1}_{\mathbf{A}}\left(n\right) \left( X_n - \mu_k \right)^2 \right) \right]^{-1}  \right),
    	\end{split}
    \end{equation}
	where $\mathbf{A} = \left\{ n | z_n = k \right\}$ and the indicator function $\mathbbm{1}_{\mathbf{A}}\left(n\right) = 1$ when $n \in A$ and $\mathbbm{1}_{\mathbf{A}}\left(n\right) = 0$ otherwise. Similarly, with the likelihood of the hyperparameters $\lambda$, $R$, $\beta$ and $W$ given in \eqref{mu_S_j_2}, and their priors given in \eqref{vague}, the conditional posteriors of the hyperparameters can be written as
	\begin{equation}
	\label{inferhyper}
	\begin{split}
	\lambda | \{ \boldsymbol{\mu}, R \}& \sim \mathcal{N} \left( \dfrac{\mu_y S_y + R \sum_{k = 1}^{K} \mu_k }{S_y + KR}, \dfrac{1}{S_y + KR} \right), \\
	R | \{\boldsymbol{\mu}, \lambda \}& \sim \mathcal{G} \left( K+1, \left[ \dfrac{1}{K + 1} \left(S_y + \sum_{k=1}^{K} \left( \mu_k - \lambda \right)^2 \right) \right]^{-1} \right), \\
	W | \{ \boldsymbol{S}, \beta \}& \sim \mathcal{G} \left( K \beta + 1, \left[ \dfrac{1}{K \beta + 1} \left(S_y + \beta \sum_{k=1}^{K} S_k \right) \right]^{-1} \right), \\
	\beta | \{ \boldsymbol{S}, W \}& \propto \Gamma \left( \dfrac{\beta}{2 }\right)^{-K} \exp \left( \dfrac{-1}{2 \beta}\right) \left( \dfrac{\beta}{2}\right)^{\left( K \beta -3 \right) / 2 } \prod_{k=1}^{K} \left( S_k W \right)^{\beta/2} \exp \left( -\dfrac{\beta S_k W}{2} \right).
	\end{split}
	\end{equation}
	Note that the exact distribution of $\beta$ is not given, but its samples can be obtained following \cite{rasmussen2000infinite}. In detail, we capture the log-concave property of $p\left( \log (\beta) | \boldsymbol{S}, W \right)$ and generate samples independently from the distribution of $\log(\beta)$. Then, we use the adaptive rejection sampling technique \cite{gilks1992adaptive} to transform these samples and obtain the value of $\beta$. 
	
	Before introducing the conditional posterior for $\boldsymbol{z}$, we note that its prior in \eqref{z_i} only suits the case where $K$ is a fixed finite parameter. To make the conditionally conjugate DPGMM applicable to the scenario with an infinite number of Gaussian components, we let $K \rightarrow \infty$ in \eqref{z_i} and the conditional prior reaches the following limits
	\begin{equation}
	p(z_n = k | \boldsymbol{z}_{-n}, \alpha) =
	\begin{cases}
	\dfrac{N_{-n,k}}{N - 1 + \alpha}, & \mbox{$N_{-n,k} > 0$}, \\
	\dfrac{\alpha}{N - 1 + \alpha}, & \mbox{$N_{-n,k} = 0$}. \\
	\end{cases}
	\end{equation}

	Now we combine the priors of $\boldsymbol{z}$ with its likelihood given in \eqref{z_i}, and the conditional posterior can be given by \eqref{two_integral}.
	
	\begin{figure*}
	\begin{equation}
	\label{two_integral}
	\begin{split}
	& p(z_n = k| \boldsymbol{z}_{-n},\mu_k, S_k, \alpha) \\
	& \propto	\begin{cases}
	p(z_n = k | \boldsymbol{z}_{-n}, \alpha) p(X_n | \mu_k, S_k,\boldsymbol{z}_{-n} ) \propto \dfrac{N_{-n,k} \mathcal{N}(X_n | \mu_k, S_k,)}{N- 1 + \alpha}, & \mbox{$N_{-n,k} > 0$}, \\
	p(z_n = k | \boldsymbol{z}_{-n}, \alpha) \int p(X_n | \mu_k, S_k) p(\mu_k, S_k | \lambda, R, \beta, W) d\mu_k d S_k, & \mbox{$N_{-n,k} = 0$}. \\
	\end{cases}
	\end{split}
	\end{equation}
	\end{figure*}

	Unfortunately, the integral in the second case of \eqref{two_integral} is not analytically tractable. Therefore, the auxiliary variable sampling algorithm \cite{doi:10.1080/10618600.2000.10474879} is employed. In detail, we add $k_0$ auxiliary components in each sampling iteration to represent the effect of the auxiliary components. Note that the posterior probability that $X_n$ belongs to the $k$-th component is proportional to $N_{-n,k}$, we use $\alpha/k_0$ for the auxiliary components and rewrite \eqref{two_integral} as
	\begin{equation}
	\label{post}
	p(z_n = k' | \boldsymbol{z}_{-n},\mu_{k'}, S_{k'}, \alpha) =
	\begin{cases}
	\dfrac{q N_{-n,k'}}{N-1+\alpha} \mathcal{N} (X_n | \mu_{k'}, S_{k'}), & \mbox{$1 \leqslant k' \leqslant K'$},\\
	\dfrac{q \alpha/k_0}{N-1+\alpha} \mathcal{N} (X_n | \mu_{k'}, S_{k'}), & \mbox{$K' < k' \leqslant K' + k_0$},
	\end{cases}
	\end{equation}
	where $k'$ is the index of unique components in each iteration during the Gibbs sampling algorithm, $q$ is the appropriate constant for normalization, and $K'$ is the number of active components. To this end, we summarize the sampling algorithm in Algorithm~\ref{Gibbs Sampling algorithm}.
	
	\begin{algorithm}
		\caption{Gibbs sampling algorithm.}  
		\label{Gibbs Sampling algorithm}   
		\begin{algorithmic}[1]
			\REQUIRE ~~ \\
			Initial observation set $\mathbf{X}$. Set a component which contains all $X_n$. Initialize the hyperparameters $\lambda$, $R$, $\beta$, and $W$, the hyperpriors $\mu_y$ and $S_y$, and the indicator set $\boldsymbol{z}$;
			\ENSURE ~~ \\
			The sets $\boldsymbol{z}$, $\boldsymbol{\mu}$, and $\boldsymbol{S}$.
			\STATE Update $\boldsymbol{\mu}$ and $\boldsymbol{S}$ conditional on the indicator $\boldsymbol{z}$ and hyperparameters $\lambda$, $R$, $\beta$ and $W$ following
			\eqref{inferpara}; \label{dddd}
			\STATE Update the hyperparameters $\lambda$, $R$, $\beta$ and $W$ conditional on hyperpriors $\mu_y$ and $S_y$ following \eqref{inferhyper}; 
			\FOR{$n = 1,2,\cdots,N$}
			\IF{$z_n \neq z_{n'}$ for all $n' \in \{1,\cdots,n-1,n+1,\cdots,N\}$} 
			\STATE Let $z_n = K' +1$. 
			\ENDIF
			\STATE Draw $\mu_{k'}$ and $S_{k'}$ for $k' \in \{ K'+1,\cdots,K'+ k_0 \}$ following \eqref{mu_S_j_2}.
			\STATE Update the indicator $z_n$ following \eqref{post}.
			\STATE Discard the empty components.
			\ENDFOR
			\label{code:fram:select}  
		\end{algorithmic}  
	\end{algorithm}  
	
	\section{Proposed Prediction-Transmission Spectrum Access Structure}
	\label{Section Spectrum}
	In this section, we propose two prediction-transmission structures (periodic and non-periodic) for spectrum access. We first introduce the functions of the prediction and transmission parts. Then, we present the details of how to determine the prediction intervals. As directly optimizing the NPLA performance is intractable, we propose to maximize an expected average utility by imposing reward (penalty) for power level match (mismatch). At last, we extend the prediction-transmission structure to an online scenario, which can handle the case where the number of the PT power levels $L$ changes after Stage I.
	
	\subsection{Functions of the Prediction and Transmission Parts}
	In the prediction part, with the inferred GMM parameters $\{ \boldsymbol{\theta},\boldsymbol{\pi},L \}$,  the ST can easily identify the current PT power level by a single sensing slot with test statistic $X_n$, based on the following criterion
	\begin{equation}
	\label{CLuster}
    k = \arg \max_{k \in \{1,\cdots, K\}} \Pr(\mathcal{H}_k | X_n) = \arg \max_{k \in \{1,\cdots, K\}} \pi_k \mathcal{N} (X_n | \mu_{k}, S_{k}).
	\end{equation}
	
	In the transmission part, the ST allocates its transmit power level $k$ to match the PT power level $l$, which means $k = l$. Note that $K$ is an estimate of $L$. In the simulation, we find that the conditionally conjugate DPGMM is able to identify $K$ ($K = L$) with a high probability\footnote{In the rare case of $K\neq L$, the NPLA performance will degrade and another round of learning will be necessary.}. When the ST matches the PT power level, it will adjust its transmit power of each level, which has been explained in Section II-B.
	
	\subsection{Periodic Structure for the Prediction-Transmission Structure}
	In periodic structure, the ST periodically predicts the PT power level in a sensing slot with duration $T_{ss}$, and then transmits for a transmission block with duration $T_{st} = \tau_s T_{ss}$. This structure can be implemented in a straightforward way, and the prediction interval remains constant whether the SR decodes the signal successfully or not.
	
	\subsection{Non-periodic Structure for the Prediction-Transmission Structure}
	\label{SectionIVC}
	For a non-periodic structure, an essential question is how to determine the prediction intervals. Basically, this needs to find out the distribution of the PT power level duration $T_{po} = \tau_p T_{ss}$, and the corresponding observation of each action. If the prediction action is taken, the observation will be the PT power level identified from the received test statistic $X_n$ according to \eqref{CLuster}. If the transmission action is taken, the observation will be a positive or negative acknowledgment (ACK) received by the ST from the SR, which indicates whether the SR correctly receives the signal from the ST. Based on these observations, the ST can infer the PT transmit power level, and then dynamically adjust the prediction intervals. As the inference may not be correct all the time, the dynamic adjustment of intervals can be formulated as a partial observable decision problem \cite{pomdp}. To address this problem, we first estimate the distribution parameter of $\tau_p$ based on $\boldsymbol{z}$. Then, we develop a reinforcement learning algorithm that correlates the observations with the current PT power level identification. 
	
	\subsubsection{PT Power Level Duration Distribution}
	Without loss of generality, the discretized PT power level duration $\tau_p$ of all hypotheses is assumed to follow a Poisson distribution with the same mean value $\nu$\footnote{If the time duration is correlated between adjacent hypotheses, the proposed algorithm also works, but at the cost of significantly increased computational complexity.}. Its cumulative distribution function can be given by
	\begin{equation}
	\label{F_X_tau}
	F_{\nu}(\tau) = \dfrac{\Gamma \left( \tau + 1,  \nu \right)  }{\Gamma \left( \tau \right) },
	\end{equation}
	where $\nu$ is the mean value of the Poisson distribution. In Stage I, we have attributed all signals $X_n$ to different components by learning the GMM parameter set $\{\boldsymbol{\theta}, \boldsymbol{\pi}, L \}$, thus we can easily obtain the samples of the PT operation durations $\tau_p$. Then, $\nu$ can be estimated through maximum likelihood estimation as
	\begin{equation}
	\nu = \dfrac{1}{M} \sum_{m = 1}^{M} \tau_{p}^m,
	\end{equation}
	where $\tau_{p}^m$ is the discretized duration of the $m$-th PT hypothesis, and $M$ is the number of PT hypotheses detected in Stage~I. 
	
	If the PT has been keeping the same power level for time $\tau$ immediately after a power level change, the probability that the PT will continue staying in the same power level during the following discretized time duration $\tau_0$ can be expressed as
	\begin{equation}
	g_{\tau}(\tau_0) = \dfrac{1-F_{\nu}(\tau+\tau_0)}{1 - F_{\nu}(\tau)}, {\ \ \ } \tau_0 \geqslant 1.
	\end{equation}
	
	\subsubsection{PT Power Transition Probability}
	In Stage I, the observations in the conditionally conjugate DPGMM are indefinite exchangeable. Therefore, the ST can not infer the PT transmit power level transition probability directly. Alternatively, we take a pragmatic approach and use the mixing proportion set $\boldsymbol{\pi}$, obtained by counting the number of observations $X_n$ in each power level, to represent the occupancy frequency of different power levels. Thus, we define an $L \times L$ transition probability matrix $\mathbf{C}$ for the PT transmit power levels. The element $C_{kj}, k,j \in \{1,\cdots,L\}$ of $\mathbf{C}$ refers to the probability that the PT transfers from the $k$-th to the $j$-th transmit power level, and is given by
	\begin{equation}
	\label{C_kf}
	C_{kj} = 
	\begin{cases} 
	\dfrac{\pi_j}{1 - \pi_k},  & \mbox{$k \neq j$}, \\
	0, & \mbox{$k = j$}.
	\end{cases}
	\end{equation}
	If the PT does not transmit with the same power level under two consecutive hypotheses then $C_{kk} = 0$. We also define the vector $\boldsymbol{c}_k$ as the $k$-th row of matrix $\mathbf{C}$.
	
	\subsubsection{Estimation Probability Matrix of the PT Power Level}
	If the PT transmits with binary power levels, the prediction performance of the ST is dictated by the detection and false alarm probabilities. However, this is no longer the case when the PT has multiple transmit power levels. Instead, we defined an $L \times L$ prediction probability matrix $\mathbf{H}$, with the element $H_{kj} = \Pr \{\mathcal{\hat{H}}_{j}|\mathcal{H}_{k}\}, k,j \in \{1,\cdots,L\}$, representing the probability that the PT is operating under hypothesis $\mathcal{H}_k$ while the detection by the ST is in favor of hypothesis $\mathcal{H}_j$. $\mathcal{\hat{H}}_j$ represents that the ST identifies the PT operating in the $j$-th transmit power level following \eqref{CLuster}. Thus the element $H_{kk}$ refers to the detection probability for hypothesis $\mathcal{H}_k$. We also define vector $\boldsymbol{h}_{j}$ as the $j$-th column of matrix $\mathbf{H}$.
	
	\subsubsection{Benefits of the ST from Prediction and Transmission Actions}
	\label{subsub_4}
    We denote the prediction action of the ST at time $\tau$ as $a_{\tau} = 0$, and its observation as $O_{\tau}^E \in \{ \mathcal{H}_k \}$. We also denote the transmission action at time $\tau$ as $a_{\tau} = 1$, and its observation as $O_{\tau}^A \in \{\mathcal{A}(\text{positive ACK}), \mathcal{N}(\text{negative ACK})\}$. It is assumed that the positive/negative ACK is returned to the ST through a dedicated feedback channel. Let $p_{\tau+\tau_0}^k$ denote the conditional probability that the PT keeps operating with the $k$-th transmit power level at time $\tau+\tau_0$ given $p_0^k = 1$ and $\{a_0,\cdots,a_{\tau},O_0,\cdots,O_{\tau}\}$, where $O_{\tau} \in \{ O_{\tau}^E, O_{\tau}^A \}$. Based on Bayesian rule, the probabilities of PT staying in the $k$-th power level at time $\tau+\tau_0, \tau_0 \in \{1,\tau_s\}$, can be given as follows.
	
	When $a_{\tau} = 0, \tau_0 = 1$, we have
	\begin{equation}
	\label{sensing}
	p_{\tau+1}^k(O^E_{\tau+1}) = \begin{cases}
	\dfrac{p_{\tau}^k g_{\tau}^E H_{kk}} {p_{\tau}^k g_{\tau}^E H_{kk} + (1 - p_{\tau}^k g_{\tau}^E)\boldsymbol{c}_{k} \boldsymbol{h}_{k}}, & \mbox{$O^E_{\tau+1} = \mathcal{H}_k$}, \\
	\dfrac{p_{\tau}^k g_{\tau}^E H_{jk}}{p_{\tau}^k g_{\tau}^E H_{jk} + (1 - p_{\tau}^k g_{\tau}^E) \boldsymbol{c}_{k} \boldsymbol{h}_{j}}, & \mbox{$O^E_{\tau+1} = \mathcal{H}_j, j \neq k$}.
	\end{cases}
	\end{equation}

	When $a_{\tau} = 1,\tau_0 = \tau_s$, we have
	\begin{equation}
	\label{transmission}
	p_{\tau+\tau_s}^k(O^A_{\tau+\tau_s}) = \begin{cases}
	\dfrac{p_{\tau}^k g_{\tau}^A}{p_{\tau}^k g_{\tau}^A + (1 - p_{\tau}^k g_{\tau}^A) \sum_{j = 1}^{L} C_{kj} p_{kj}^{ACK}}, & \mbox{$O^A_{\tau+\tau_s}$ = $\mathcal{A}$}, \\
	0, & \mbox{$O^A_{\tau+\tau_s}$ = $\mathcal{N}$,}
	\end{cases}
	\end{equation}
	where we assume the positive and negative ACKs from the SR can be received by the ST error free. Besides, $p_{kj}^{ACK}$ denotes the probability that the SR decodes signals correctly when ST is on the $k$-th level and PT is one the $j$-th level. For simplicity, we set $p_{kj}^{ACK} = 1$ when $k > j$ and $p_{kj}^{ACK} = 0$ otherwise. In \eqref{sensing} and \eqref{transmission}, $g_{\tau}^E =g_{\tau}(1)$ and $g_{\tau}^A = g_{\tau}(\tau_s)$, where the superscripts $E$ and $A$ represent the prediction and transmission, respectively. The expected utility that the ST obtains at time $\tau$ with the PT operating with the $k$-th transmit power level, which is $r(p_{\tau}^k,a_{\tau},k)$, can be given by
	\begin{equation}
	\label{ref_transmission}
	r(p_{\tau}^k,a_{\tau},k) = \begin{cases}
	\left[ p_{\tau}^k g_{\tau}^A D_k - (1 - p_{\tau}^k g_{\tau}^A) \sum_{j = 1}^{L} C_{kj} Y_j \right] \tau_s, & \mbox{$a_{\tau} = 1$}, \\
	0, & \mbox{$a_{\tau} = 0$}.
	\end{cases}
	\end{equation}
	Hereafter, $k$ denotes the ST transmit power level, $l$ is the PT transmit power level, and $K = L$. In \eqref{ref_transmission}, $D_k$ is the reward that the ST will receive when the ST transmits with the $k$-th power level and $k = l$. Meanwhile, $Y_k$ is the penalty that the ST will receive when the ST transmits with the $k$-th power and $k \neq l$.
	
	\subsubsection{Prediction-Transmission Structure Optimization}
	An ST access policy $\epsilon = \left[ d_0, \cdots ,d_{\tau}, \cdots \right] $ maps the ST belief space $\{p_{\tau}^k, \tau \geqslant 0 \}$ to the action space $\{ a_{\tau}, \tau \geqslant 0 \} $. Thus, the optimal prediction-transmission
	policy aims to maximize the expected average utility, which can be given by
	\begin{equation}
	\label{max_lim}
	\max_{\epsilon} \lim_{M' \rightarrow \infty} \dfrac{\sum_{m=M+1}^{M'} \sum_{{\tau}=0}^{\tau_p^m - 1} r(p_{\tau}^k,a_{\tau},k) a_{\tau} / (M' - M)}{\sum_{m=M+1}^{M'} \tau_p^m / (M' - M)}.
	\end{equation}
	Note that if $a_{\tau} = 1$ and $\tau_s \geqslant 2$, $a_{\tau + \tau_0}, \tau_0 = \{1,\cdots,\tau_s - 1\}$ is not defined in Section \ref{subsub_4}. This is because the ST keeps transmitting between time $\tau$ and $\tau + \tau_s$, thus no action is taken during this period. Therefore, we assign $a_{\tau + \tau_0} = 0, \tau_0 = \{1,\cdots,\tau_s - 1\}$ if $a_{\tau} = 1$ and $\tau_s \geqslant 2$ in the action space $\{a_{\tau}, \tau \geqslant0 \}$. The total utility obtained by the ST during each PT hypothesis is i.i.d.. Thus, by the law of large numbers, the maximization problem in \eqref{max_lim} can be rewritten as
	\begin{equation}
	\max_{\epsilon} \mathrm{E} \left[ \sum_{\tau=0}^{ \tau_p - 1 } r(p_{\tau}^k,a_{\tau},k) a_{\tau} \right].
	\end{equation}
	In other words, instead of maximizing the average expected utility, we translate the problem to the maximization of the utility in each PT hypothesis. For convenience, let $V_{\epsilon}(0,p_{\tau}^k,k)$ denote the expected utility that can be achieved in each PT hypothesis following policy $\epsilon$, which is
	\begin{equation}
	V_{\epsilon}({\tau} = 0,p_{\tau}^k = 1,k) = E_{\epsilon} \left[ \sum_{{\tau}=0}^{\tau_p - 1}r(p_{\tau}^k,a_{\tau},k) a_{\tau} \right].
	\end{equation}
	Then, we can define the maximum utility that can be achieved by the ST in each PT hypothesis as
	\begin{equation}
	\label{V_0}
	V({\tau} = 0,p_{\tau}^k = 1,k) = \sup_{\epsilon} V_{\epsilon}({\tau} = 0,p_{\tau}^k = 1,k).
	\end{equation}
	In \eqref{V_0}, $V({\tau},p_{\tau}^k,k)$ is directly determined by the action taken at time ${\tau}$, thus it can be expressed as
	\begin{equation}
	\label{optimality equation}
	V({\tau},p_{\tau}^k,k) = \max \left\lbrace E(\tau,p_{\tau}^k,k), A(\tau,p_{\tau}^k,k)\right\rbrace ,
	\end{equation}
	where $E(\tau,p_{\tau}^k,k)$ and $A(\tau,p_{\tau}^k,k)$ are the expected utilities that can be obtained by the ST through prediction and transmission, respectively. We have
	\begin{equation}
	\label{utility_L}
	E(\tau,p_{\tau}^k,k) = \sum_{j = 1}^{L} \Pr\left[ O_{\tau + 1}^E = \mathcal{H}_j \right] V(\tau + 1,p_{\tau + 1}^k(\mathcal{H}_j),k),
	\end{equation}
	and
	\begin{equation}
	\label{utility_M}
	\begin{split}
	A(\tau,p_{\tau}^k,k) = & \sum_{\mathcal{J} \in \{\mathcal{A,N}\}} \Pr\left[ O_{\tau + \tau_s}^A = \mathcal{J} \right] V(\tau + \tau_s, p_{\tau + \tau_s}^k(\mathcal{J}),k) + r(p_{\tau}^k,1,k) \\
	= & \Pr \left[ O_{\tau + \tau_s}^A = \mathcal{A} \right] V(\tau + \tau_s,p_{\tau + \tau_s}^k(\mathcal{A}),k) + r(p_{\tau}^k,1,k).
	\end{split}
	\end{equation}
	
	\begin{lemma}
		\label{lemma_1}
		$V(\tau,p_{\tau}^k,k)$ is a convex function of $p_{\tau}^k$ for given $\tau$ and $k$.
	\end{lemma}
	The proof of the above lemma can be easily obtained following \cite[page 58-59]{ross2014introduction}, and is omitted here. 
	
	It is clear that $V(\tau,p_{\tau}^k,k)$ is derived backward in time domain in \eqref{optimality equation}, \eqref{utility_L} and \eqref{utility_M}. Thus it will be helpful for the derivation of the optimal policy if an upper bound of $\tau$ can be established, which is given by
	\begin{equation}
	\label{T*}
	\mathcal{T}_k = \min\left\lbrace \tau': g_{\tau}^A < \frac{\sum_{j = 1}^{L} C_{kj} Y_j}{ \sum_{j = 1}^{L} C_{kj} Y_j + D_k }, \forall \tau > \tau' \right\rbrace.
	\end{equation}
	In each PT hypothesis with the $k$-th transmit power level, the transmission action will not be taken by the ST after $\mathcal{T}_k$. This is because, $\forall \tau > \mathcal{T}_k, r(1,1,k) < 0$, which means the ST will certainly receive negative reward if it transmits after $\mathcal{T}_k$, even if the PT is estimated as idle at time $\tau$. Therefore, $V(\tau,1,k) = 0, \forall \tau>\mathcal{T}_k$. When it comes to the range of $\mathcal{T}_k$, as $\tau_p$ follows a Poisson distribution in \eqref{F_X_tau}, we find that $\mathcal{T}_k < \infty$ always holds according to \cite{5062155}.
	
	\begin{lemma}
		\label{lemma_2}
		The optimal utility function $V({\tau},p_{\tau}^k,k)$ increases with $p_{\tau}^k$ for given $\tau$ and $k$.
	\end{lemma}
	\begin{proof}
		The proof is given in Appendix \ref{Proof}.
	\end{proof}
	
	\begin{lemma}
		\label{lemma_3}
		$E(\tau,p_{\tau}^k,k)$ and $A(\tau,p_{\tau}^k,k)$ are convex functions increasing with $p_{\tau}^k$ for given ${\tau}$ and $k$.
	\end{lemma}
	\begin{proof}
		We have proved that $E(\tau,p_{\tau}^k,k)$ and $A(\tau,p_{\tau}^k,k)$ increase with $p_{\tau}^k$ in Lemma \ref{lemma_2}. Next, to prove their convexity, we derive their second order derivatives with regard to $p_{\tau}^k$. Combining \eqref{sensing} and \eqref{utility_L}, the second order derivative of $E(\tau,p_{\tau}^k,k)$ can be given by
		\begin{equation}
		\frac{dE^2(\tau,p_{\tau}^k,k)}{d^2 p_{\tau}^k} = \sum_{j = 1}^{L} \frac{\left( g_{\tau}^E H_{jk} \boldsymbol{c}_{k} \boldsymbol{h}_{j} \right)^2 V''(\tau+1,p_{\tau + 1}^k(\mathcal{H}_j),k)}{\left[  p_{\tau}^k g_{\tau}^E H_{jk} + (1 - p_{\tau}^k g_{\tau}^E) \boldsymbol{c}_{k} \boldsymbol{h}_{j} \right] ^3 },
		\end{equation}
		which is positive as $V({\tau},p_{\tau}^k,k)$ is convex. The second derivative of $A(\tau,p_{\tau}^k,k)$ can be proved positive similarly. Therefore, we complete the proof.
	\end{proof}
	
	We note that $E(\tau,0,k) = 0$ in \eqref{utility_L} and $A(\tau,0,k) = -\sum_{j = 1}^{L} C_{kj} Y_j \tau_s < 0$ in \eqref{utility_M}. Combining Lemma \ref{lemma_3}, we define the probability thresholds as
	\begin{equation}
	\label{p_tk}
	p_{\tau}^{k*} = \min\{ p_{\tau}^k:E(\tau,p_{\tau}^k,k) \leqslant A(\tau,p_{\tau}^k,k) \},
	\end{equation}
	and
	\begin{equation}
	p_{\tau}^{k**} = \max\{ p_{\tau}^k:E(\tau,p_{\tau}^k,k) \leqslant A(\tau,p_{\tau}^k,k) \}.
	\end{equation}
	Then, we can give the following optimized protocol, where $a_{\tau}^*$ is the optimal action at time $\tau$.
	\begin{theorem}
		\label{theorem_1}
		If $p_{\tau}^{k*} = p_{\tau}^{k**}$, we have
		\begin{equation}
		a_{\tau}^* = 
		\begin{cases}
		\label{case 1}
		0, & \mbox{$p_{\tau}^k \leqslant p_{\tau}^{k*}$},\\
		1, & \mbox{others},
		\end{cases}
		\end{equation}
		and if $p_{\tau}^{*} < p_{\tau}^{**}$, we have
		\begin{equation}
		a_{\tau}^* = 
		\begin{cases}
		\label{case 2}
		0, & \mbox{$p_{\tau}^k \leqslant p_{\tau}^{k*}$ or $p_{\tau}^k \geqslant p_{\tau}^{**}$ }, \\
		1, & \mbox{others}.
		\end{cases}
		\end{equation}
    \end{theorem}
	The two cases are illustrated in Fig. \ref{illustration}. When $p_{\tau}^{k**}=p_{\tau}^{k*}$, $E(\tau,p_{\tau}^k,k)$ and $A(\tau,p_{\tau}^k,k)$ can be illustrated in Fig. \ref{illustration} (a) following Lemma \ref{lemma_3}. In this case, when $p_{\tau}^k \leqslant p_{\tau}^{k*}$, we get $E(\tau,p_{\tau}^k,k) \geqslant (\tau,p_{\tau}^k,k)$. The ST will choose the prediction action for larger expected utility ($a_{\tau}^*=0$). When $p_{\tau}^k>p_{\tau}^{k*}$, we get $a_{\tau}^*=1$. Note that the ST will choose the next action at the end of each transmission or prediction period. Fig. \ref{illustration} (b) follows a similar interpretation.
	\begin{figure}
		\centering
		\includegraphics[width=0.9\textwidth]{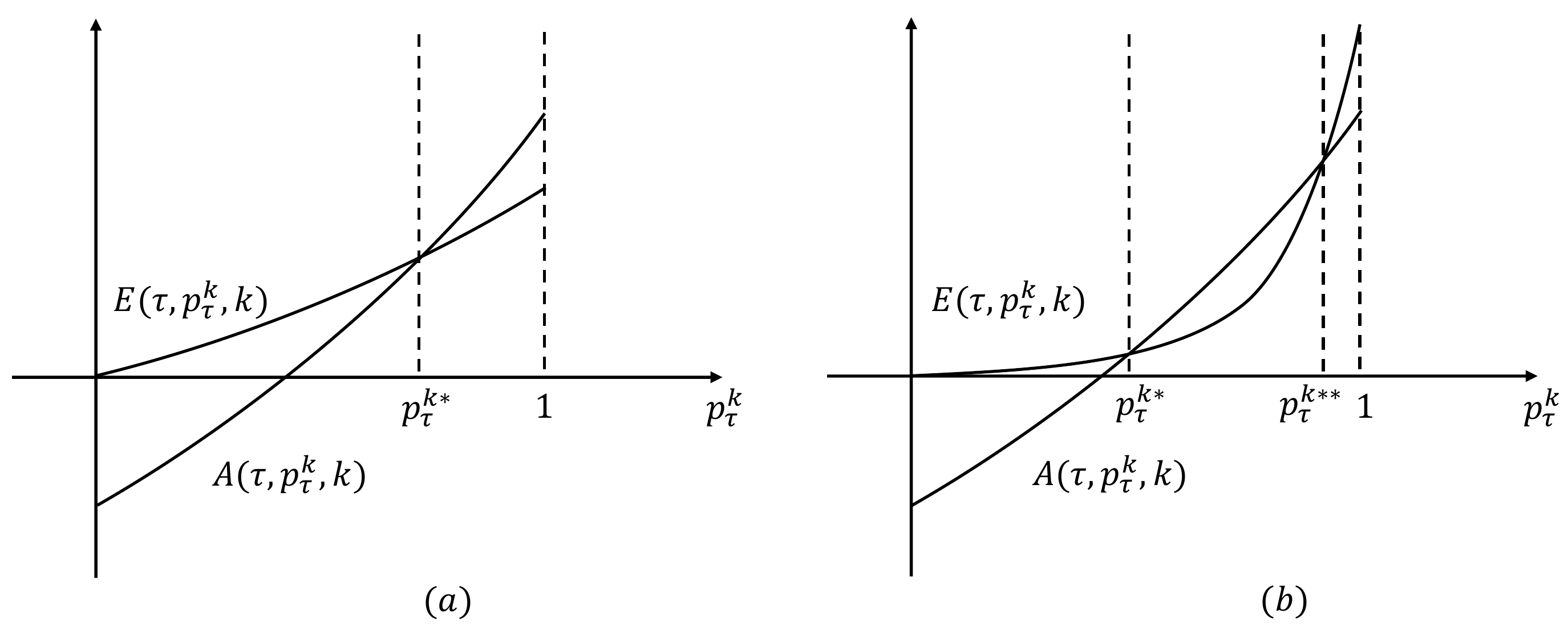}
		\caption{The illustration of $A(\tau,p_{\tau}^k,k)$ and $E(\tau,p_{\tau}^k,k)$ versus $p_{\tau}^k$, for ($a$) $p_{\tau}^{k*} = p_{\tau}^{k**}$ and ($b$) $p_{\tau}^{k*} < p_{\tau}^{k**}$.}
		\label{illustration}
	\end{figure}
	
	\subsection{Prediction-Transmission in an Online Scenario}
	In Stage II, our proposed prediction-transmission structure can dynamically adjust the intervals between two sensing slots to increase the secondary network throughput. We note that this structure only works when the number of PT transmit power levels $L$ remains unchanged. This is because the prediction-transmission structure is based on an one-off learning period in Stage I, and the learning results will become obsolete if $L$ changes after Stage I. However, this change is likely to happen in the long run as a consequence of the PT adapting to the environment fluctuation or QoS variation. This will render the proposed prediction-transmission structure less effective. To rectify this situation, we take advantage of the fact that the ST still receives the PT signals from time to time, reflecting the already changed PT power levels. We will explore these signals, and extend the spectrum access method to an online scenario to accommodate this highly dynamic PT power level characteristics.
	
	In the online prediction-transmission structure, the ST not only identifies the PT power level in the sensing slot, but also stores the received PT signals ($X_n^*$). Then, it updates its observation set, starting from $\textbf{X}$ from Stage I, through replacing the old signals with newly-received ones, while the size of $\textbf{X}$ keeps the same. For example, after receiving the first signal $X_1^*$ in Stage II, $\textbf{X}$ will be updated as
	\begin{equation}
	\boldsymbol{X} = \{ X_2,...,X_{N},X_1^* \}.
	\end{equation}
	To make online update of $\{ \boldsymbol{\theta},\boldsymbol{\pi},L \}$, we can revise Algorithm \ref{Gibbs Sampling algorithm} by adding the following algorithm flow before line \ref{dddd}.
	\begin{algorithm}
		\caption{Algorithm Flow}  
		\begin{algorithmic}
			\FOR {$n = 2,\cdots,N$}
			\STATE Let $X_{n-1} = X_n$, $z_{n - 1} = z_n$.
			\ENDFOR
			\STATE Let $X_N = X_1^*$.
			\STATE Update the indicator $z_n$ following \eqref{post}.
		\end{algorithmic}  
	\end{algorithm}
	
	\section{Numerical Results}
	\label{Section Numerial Results}
	In this section, numerical results are provided to illustrate the advantages of the proposed two-stage spectrum sharing strategy. We first evaluate the performance of the prediction part in Stage II, in terms of different learning methods in Stage I including both parametric and nonparametric. Then, we numerically verify the theoretical results of the optimal thresholds in Theorem \ref{theorem_1} in the prediction-transmission structure. Finally, we illustrate the superiority of the proposed spectrum sharing strategy on an overall basis in comparison with other methods, which are different from the proposed one in Stage I or II. In addition, we demonstrate the advantages of the proposed online prediction-transmission structure in handling the highly dynamic scenario where the number of PT power levels changes after Stage I.
	
	In the simulation, we set the power level $L=4$ and the probability of each hypothesis $\Pr(\mathcal{H}_l) = 0.25$. It is assumed that the noise variance $\sigma_u^2 = 1$, the PT transmit powers $P_{PT,1}:P_{PT,2}:P_{PT,3} = 1:2:3$, and $P_{PT,4} = 0$. Recall that $\gamma_{st}^l = P_{PT,l}/\sigma_u^2$ is the received SNR at the ST, thus we have $\gamma_{st}^1:\gamma_{st}^2:\gamma_{st}^3 = 1:2:3$ and $\gamma_{st}^4 = 0$. The average SNR at the ST is defined as $\gamma_{st} = (1/L)\sum_{l = 1}^{L} \gamma_{st}^l$. In addition, we set the time duration of a sensing slot $T_{ss}= 2$ ms, the reward $D_k = 1$, and penalty $Y_l = 1$ in \eqref{ref_transmission}. 
	
	\subsection{Performance of the PT Power Level Identification}
	For our proposed strategy, in the prediction part of Stage II, the ST can easily identify the current PT power level based on the inferred $\{ \boldsymbol{\theta},\boldsymbol{\pi},L \}$ in Stage I and $X_n$, $n > N$. Note that in addition to the proposed conditionally conjugate DPGMM, there are other learning methods to obtain this parameter set, such as expectation maximization GMM (EMGMM) \cite{bilmes1998gentle}, mean shift \cite{1000236}, kernel density estimation (KDE), and density based spatial clustering of application with noise (DBSCAN). Among them, the EMGMM is a parametric clustering method requiring the prior knowledge of $L$, while the others belong to a nonparametric class without the need for such prior knowledge. Taking into account the multiple power levels, we define the probability of correct PT power level prediction as
	\begin{equation}
	P_c = \sum_{l = 1}^{L} \Pr \{\mathcal{\hat{H}}_{l}|\mathcal{H}_{l}\} \Pr \{ \mathcal{H}_l \}.
	\end{equation}
	
	\begin{figure*}[ht]
		\centering
		\subfigure[$N^{(s)} = 10^4$.]{
			\includegraphics[width=0.45\textwidth]{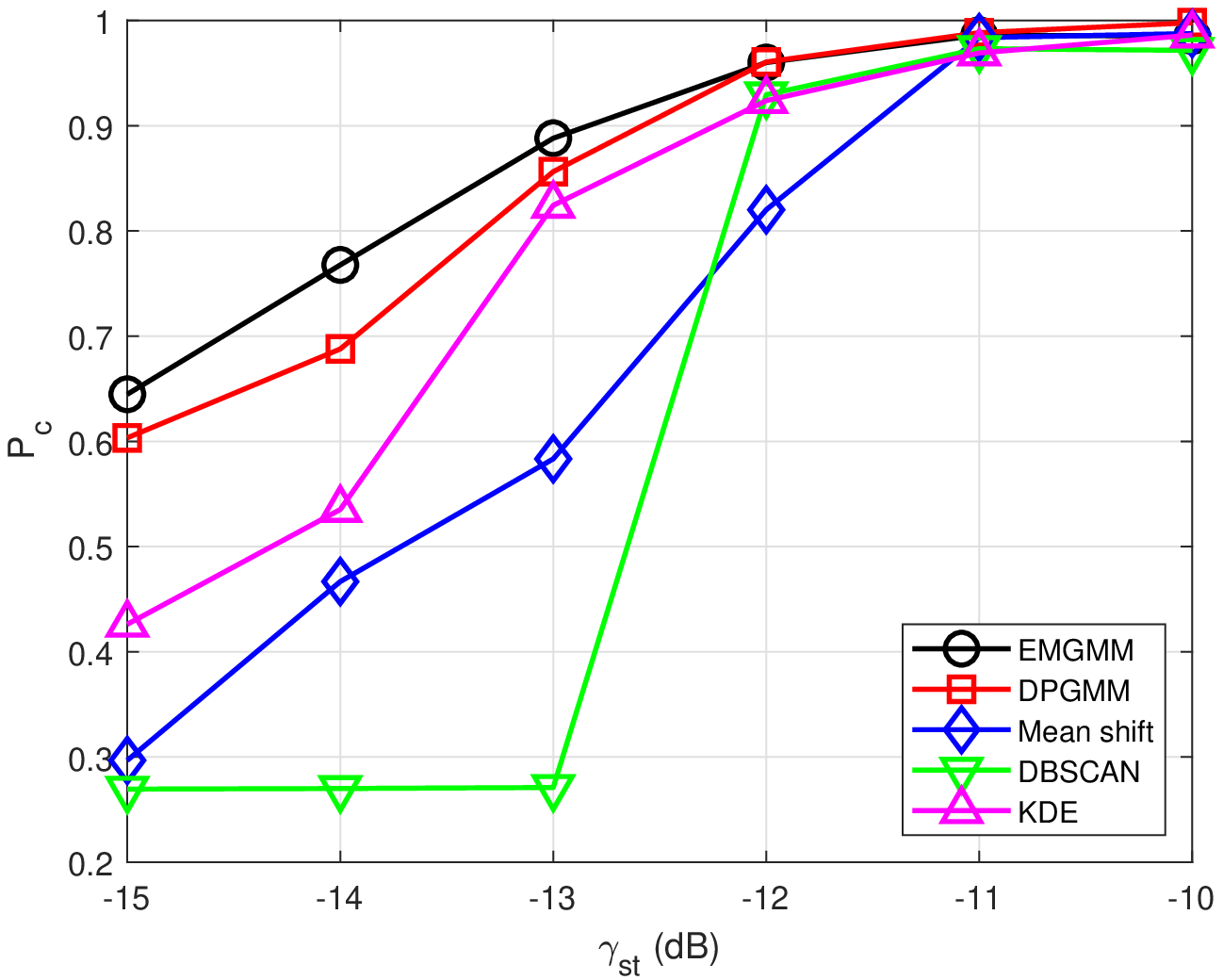}
			\label{cluster accuracy}
		}
		\subfigure[$\gamma_{st} = -12$ dB.]{
			\includegraphics[width=0.45\textwidth]{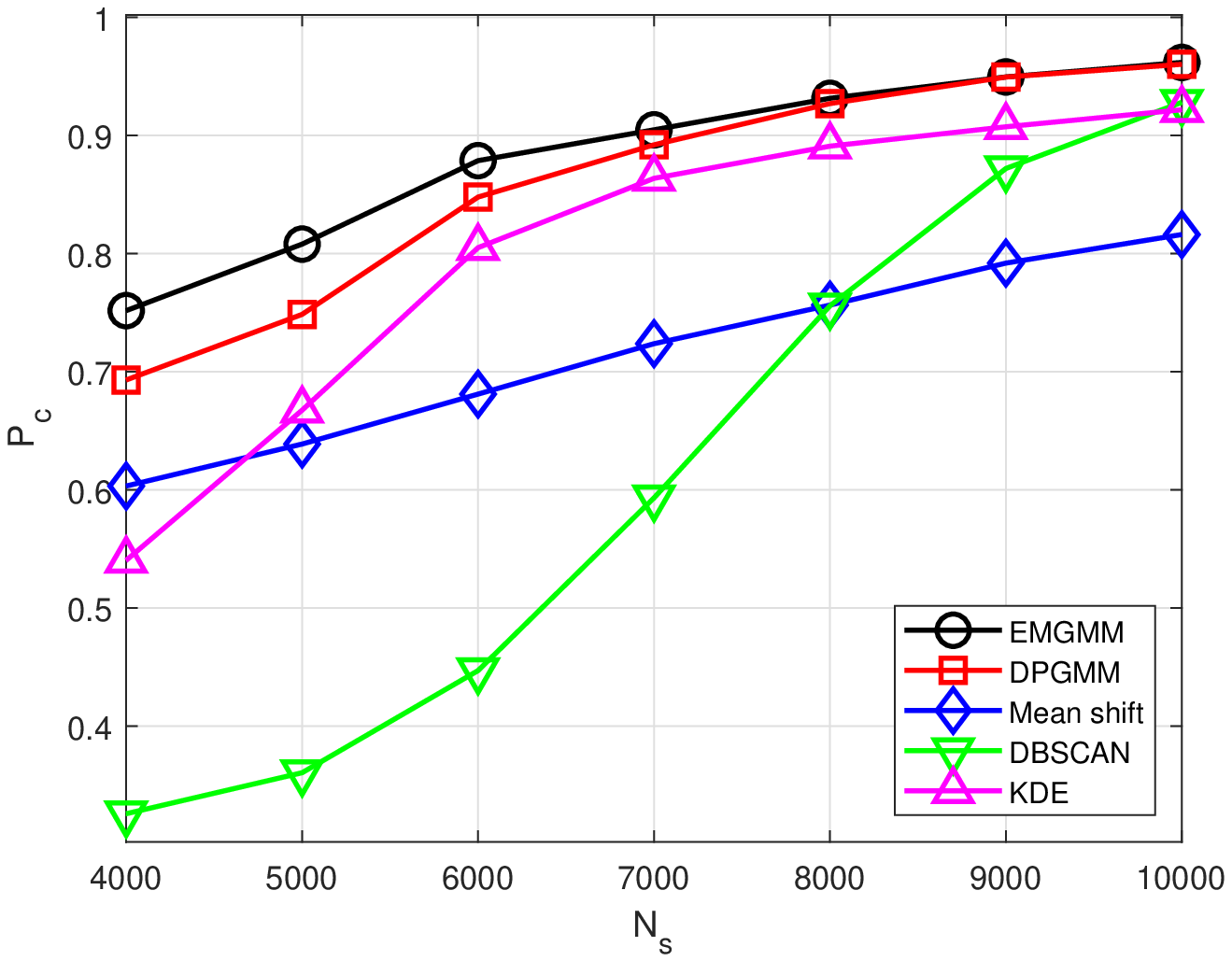}
			\label{power estimation}
		}
		\caption{The probability of correct PT power level prediction in Stage II ($P_c$) versus $\gamma_{st}$ and $N^{(s)}$.}
		\label{combined picutre 1}
	\end{figure*}
	
	In Fig. \ref{combined picutre 1}, we illustrate $P_c$ for five different learning methods, with respect to different $\gamma_{st}$ (Fig. \ref{cluster accuracy}) and $N^{(s)}$ (Fig. \ref{power estimation}). It is shown in Fig. \ref{cluster accuracy} that, in general, $P_c$ increases with $\gamma_{st}$ for all methods. This is because the gap between the adjacent transmit powers increases with $\gamma_{st}$, rendering them more distinguishable from the perspective of machine learning. In Fig. \ref{power estimation}, $P_c$ improves with increasing $N^{(s)}$ for all methods, because a larger $N^{(s)}$ results in a smaller variance of each Gaussian distribution in the mixture model. In Fig. \ref{combined picutre 1}, without the prior knowledge of $L$, the proposed DPGMM significantly outperforms mean shift and DBSCAN, particularly for small $\gamma_{st}$. The DPGMM also achieves larger $P_c$ than KDE. Meanwhile, despite the lack of the prior knowledge of $L$, DPGMM is only slightly inferior to EMGMM, and the gap becomes negligible for increased $\gamma_{st}$ or $N^{(s)}$.
	
	\subsection{Theoretical Verification of the Prediction-Transmission Structure}
	In this subsection, we first illustrate the optimized thresholds through numerical search. Note that only the case in Fig. \ref{illustration} (a) will occur in our scenario. Then, we compare the maximized utility $V(0,1,k)$ in \eqref{V_0} with numerical results. 
	
	\begin{figure*}[ht]
		\centering
		\subfigure[The optimized threshold $p_{\tau}^{k*}$ versus $\tau$ with $\tau_s = 4$.]{
			\includegraphics[width=0.45\textwidth]{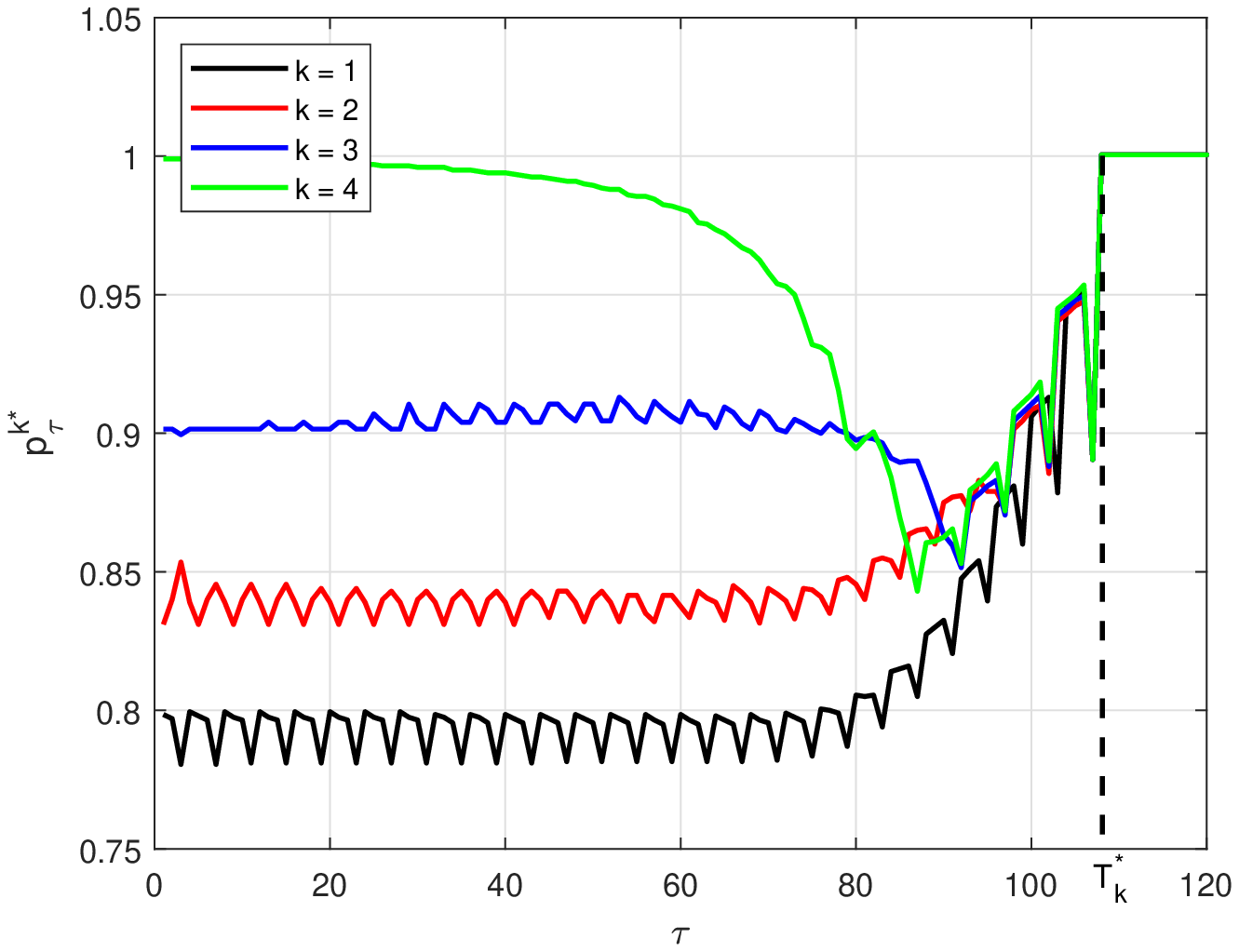}
			\label{optimized p l}
		}
		\subfigure[$V(0,1,k)$ of theoretical and numerical results versus $\tau_s$.]{
			\includegraphics[width=0.45\textwidth]{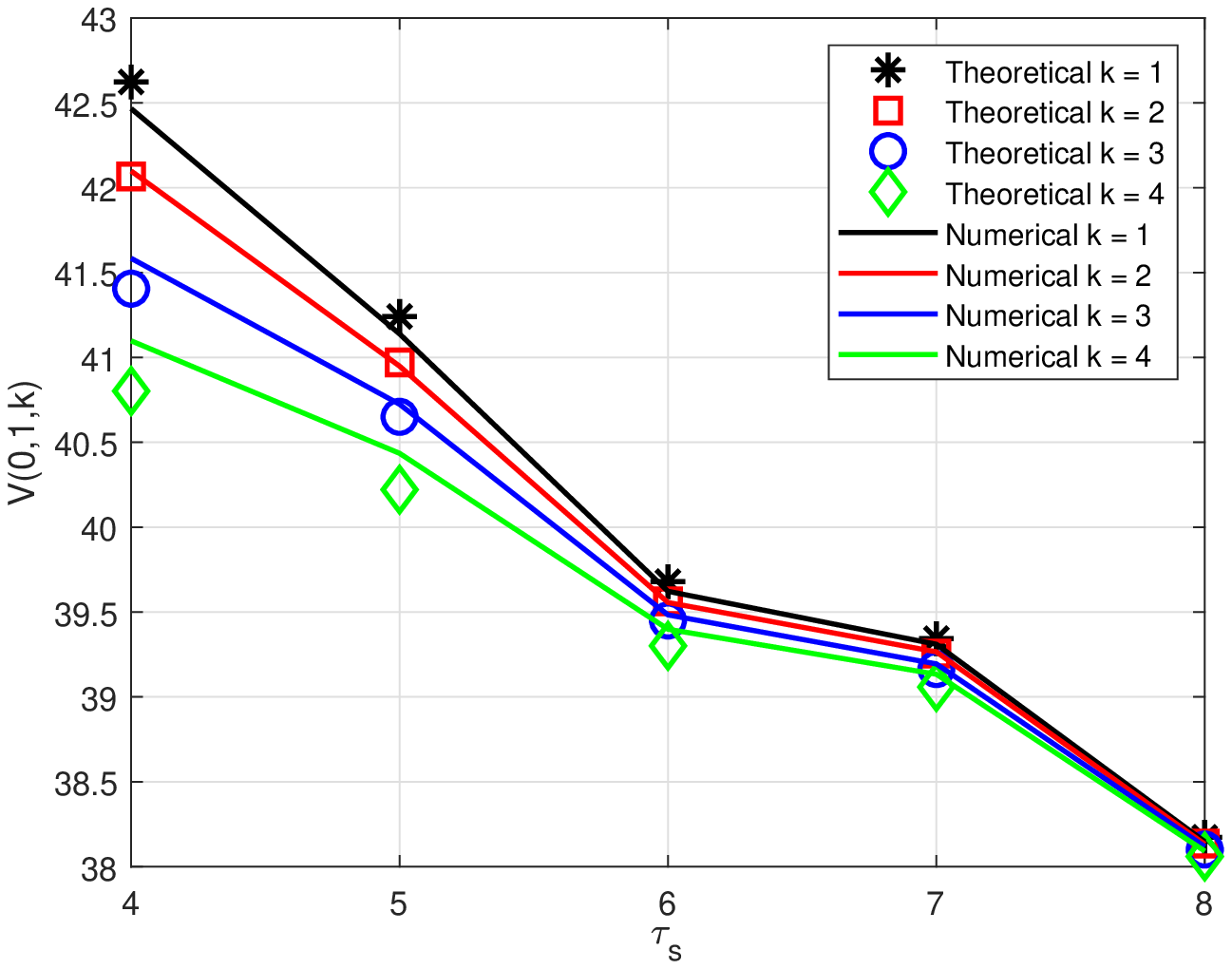}
			\label{optimized p l_mix}
		}
		\caption{Theoretical verification of the prediction-transmission structure with $f_s = 5$ MHz, $N^{(s)} = 10^4$, $\gamma_{st} = -12$ dB, and $\nu = 100$.}
		\label{combinded picture 2}
	\end{figure*}
	
	In Fig. \ref{optimized p l}, we illustrate $p_{\tau}^{k*}$ in \eqref{case 1} for different PT transmit power levels identified by the ST. It is shown that $p_{\tau}^{k*}$ is time varying, and is larger for higher PT transmit power level, which indicates that the ST favors a prediction action for a higher transmit power level $k$. However, it is seen that $p_{\tau}^{k*}$ becomes 1 when the time exceeds $\mathcal{T}_k$, which means that the ST will not transmit after $\mathcal{T}_k$ regardless of the value of $p_{\tau}^{k}$. This observation is consistent with \eqref{T*}. In Fig. \ref{optimized p l_mix}, we compare the maximized utility $V(0,p_{\tau}^{k*},k)$ in \eqref{V_0} for different $k$. It is observed that the numerical results match the theoretical ones well, which verifies the effectiveness of the proposed reinforcement learning method. It is seen that $V(0,1,k)$ decreases with $\tau_s$ for all $k$. This is because the increase in $\tau_s$ renders the prediction of $l$ in the next transmission block less accurate. Thus, the ST is more likely to carry out prediction when the current PT power level duration is close to the end. Naturally, this decreases the utility.
	
	\subsection{Advantages of the Conditionally Conjugate DPGMM in Stage I}
	\label{First Stage}
	In this subsection, we illustrate the impact of different learning methods in Stage I on the NPLA performance of the proposed spectrum sharing strategy. The NPLA performance is defined as
	\begin{equation}
	U(\tau) =  \dfrac{\tau_s}{\tau} \sum_{\tau_0 = 0}^{ \tau } a_{\tau_0} \psi( \sum_{\tau_1 = 0}^{\tau_s - 1} | k_{\tau_0 + \tau_1} - l_{\tau_0 + \tau_1} |),	\end{equation}
	where $k_{\tau}$ and $l_{\tau}$ denote the ST and PT transmit power levels at time $\tau$, respectively. We have $\psi(x) = 1$ when $x = 0$, and $\psi(x) = 0$ otherwise. Basically, a larger $U(\tau)$ leads to a better tradeoff between the secondary network throughput and the interference to the primary network.
	
	\begin{figure*}[ht]
		\centering
		\subfigure[Periodic structure for the prediction-transmission structure.]{
			\includegraphics[width=0.45\textwidth]{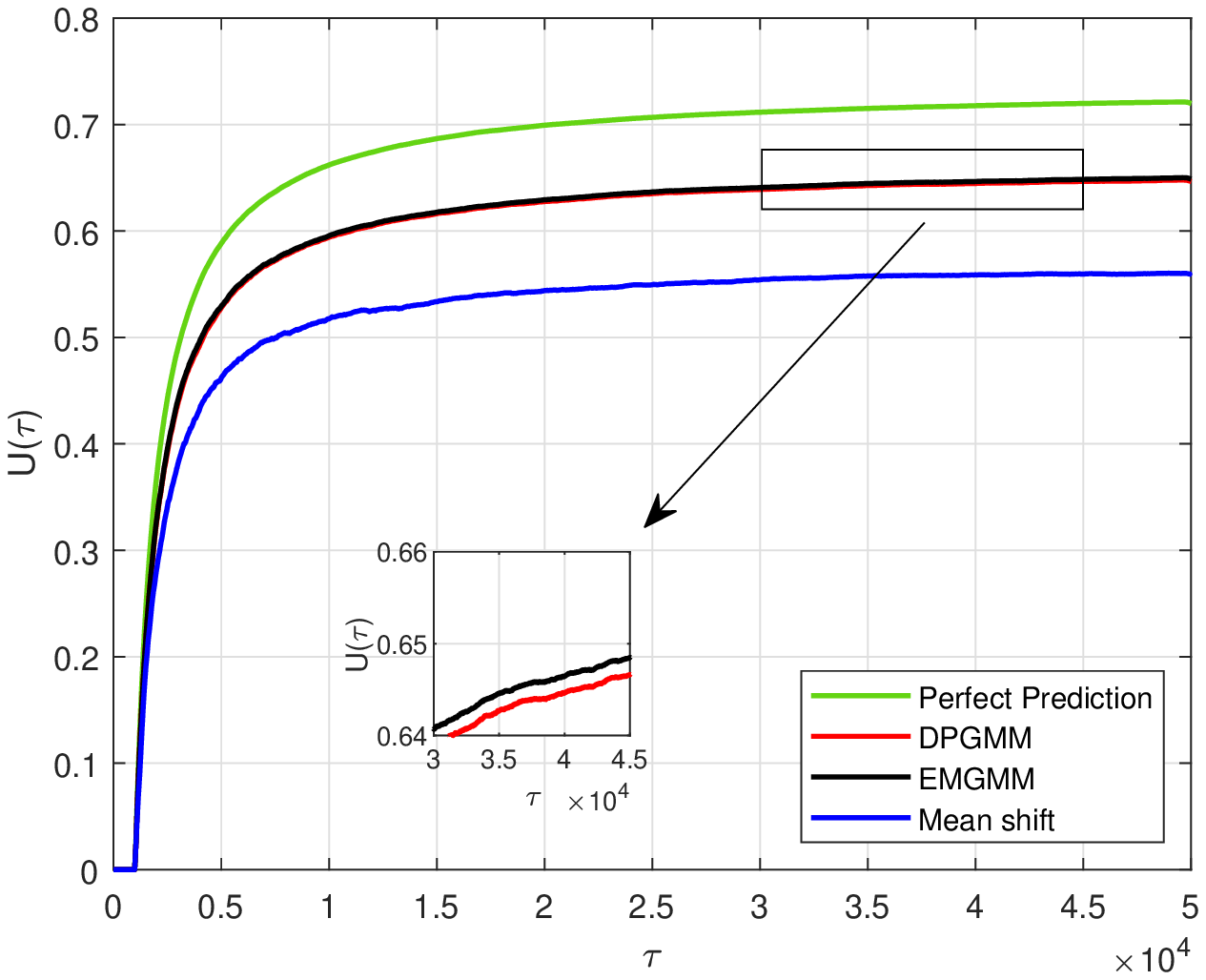}
			\label{Utility_Step_1_1}
		}
		\subfigure[Non-periodic structure for the prediction-transmission structure.]{
			\includegraphics[width=0.45\textwidth]{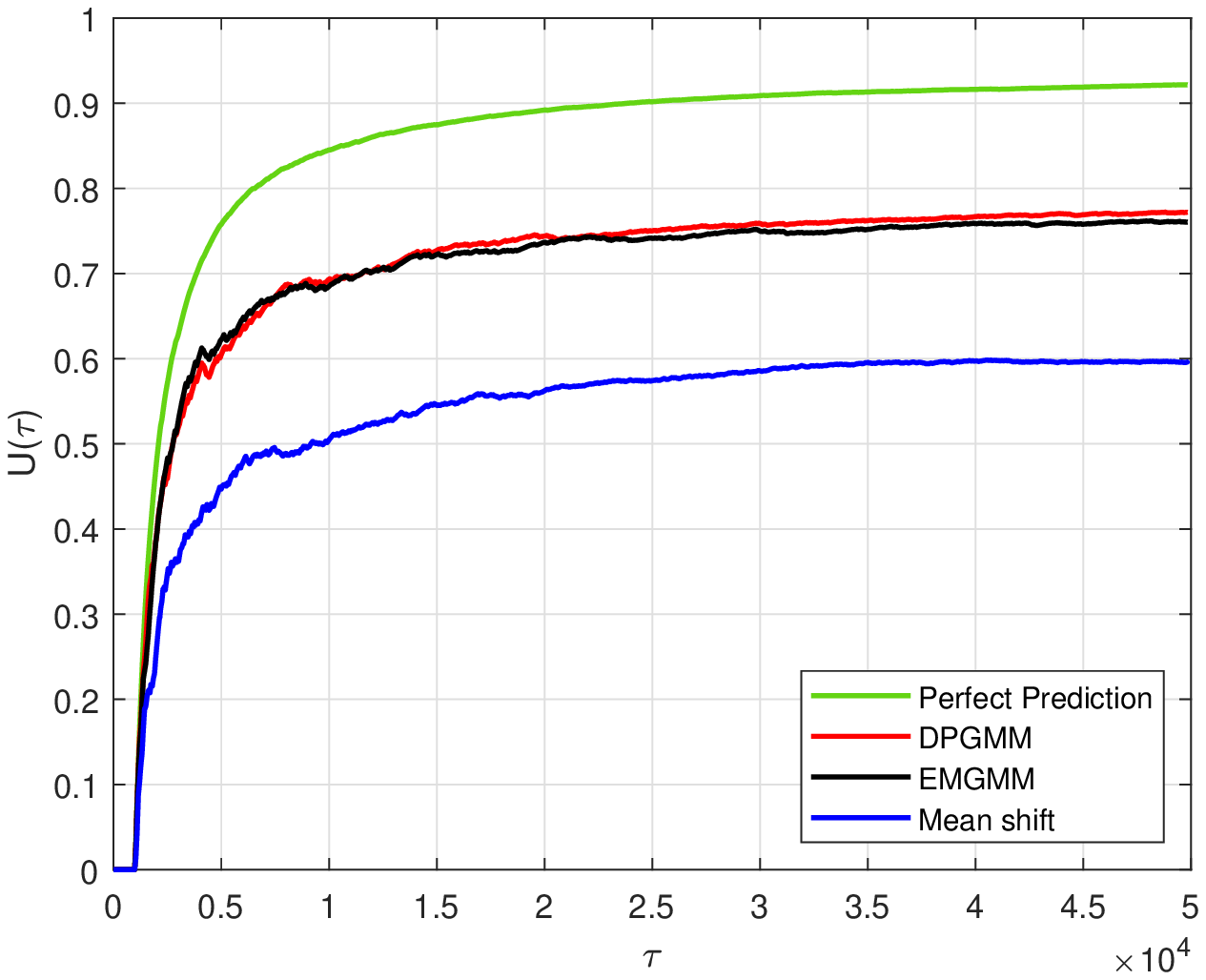}
			\label{Utility_Step_1}
		}
		\caption{The NPLA performance $U(\tau)$ for the spectrum sharing strategies using different clustering methods in Stage I with $f_s = 2.5$ MHz, $N^{(s)} = 5\times 10^3$, $\tau_s = 4$, $\gamma_{st} = -12$ dB, and $\nu = 50$.}
		\label{first_stage_both}
	\end{figure*}
	
	In Fig. \ref{first_stage_both}, we compare $U(\tau)$ for different learning algorithms, considering both periodic (Fig. \ref{Utility_Step_1_1}) and non-periodic structures (Fig. \ref{Utility_Step_1}) in the prediction part of Stage II. In addition to the three learning methods in Fig. \ref{combined picutre 1} (conditionally conjugate DPGMM, EMGMM and mean shift), we include for reference an upper bound with perfect prediction. The same $N$, $T_{ss}$ and $\tau_s $ are assumed for all the methods.
	
	In Fig. \ref{Utility_Step_1_1}, it is shown that $U(\tau)$ remains 0 when $\tau \leqslant 10^3$ due to the learning period (i.e., no transmission), and begins to increase when the ST goes into Stage II ($\tau > 10^3$). $U(\tau)$ refers to the average utility from the beginning of Stage I to time $\tau$, will keep stable in the long term. Thus it first increases significantly and asymptotically approach the constant after certain time point, which is $5 \times 10^3$ shown in the figure. With reference to Fig. \ref{power estimation}, we can see that the learning method with better $P_c$ corresponds to higher $U(\tau)$. This clearly collaborates the significance of the prediction accuracy. Similar conclusions can be drawn in Fig. \ref{Utility_Step_1} for the non-periodic structure.
		
	\subsection{Advantages of the Prediction-Transmission Structure in Stage II}
	In Fig. \ref{Second_stage}, we compare $U(\tau)$ of the proposed spectrum sharing strategy with periodic and non-periodic structures in Stage II for different $\tau$ and $\tau_s$. Note that conditionally conjugate DPGMM is used in Stage I for all the cases. As an upper bound, we include a perfect system where the ST can always accurately track the PT power level variation. 
	
	Similar to Fig. \ref{first_stage_both}, it is shown in Fig. \ref{Second_stage_1} that $U(\tau)$ of three different structures remains 0 when $\tau \leqslant 10^3$, and approaches certain positive value when $\tau > 10^3$. This comes from a similar reason explained in Section \ref{First Stage}. Fig. \ref{Second_stage_1} shows that the non-periodic structure outperforms the periodic one, which verifies the benefit of dynamically adjusting the prediction~intervals.
	
	Fig. \ref{Second_stage_2} shows that $U(\tau)$ of three different structures versus $\tau_s$ when $\tau = 5 \times 10^4$. We find that when $\tau_s$ increases, $U(\tau)$ of the periodic structure increases. This is because a larger $\tau_s$ results in a smaller time proportion of prediction, and it will benefit the periodic strategy. By contrast, for non-periodic structure and the perfect system, $U(\tau)$ will increase with a smaller $\tau_s$. This is because for a smaller $\tau_s$, the ST is more flexible in transmission block allocation.
	
	\begin{figure*}[ht]
		\centering
		\subfigure[$\tau_s = 2$.]{
			\includegraphics[width=0.45\textwidth]{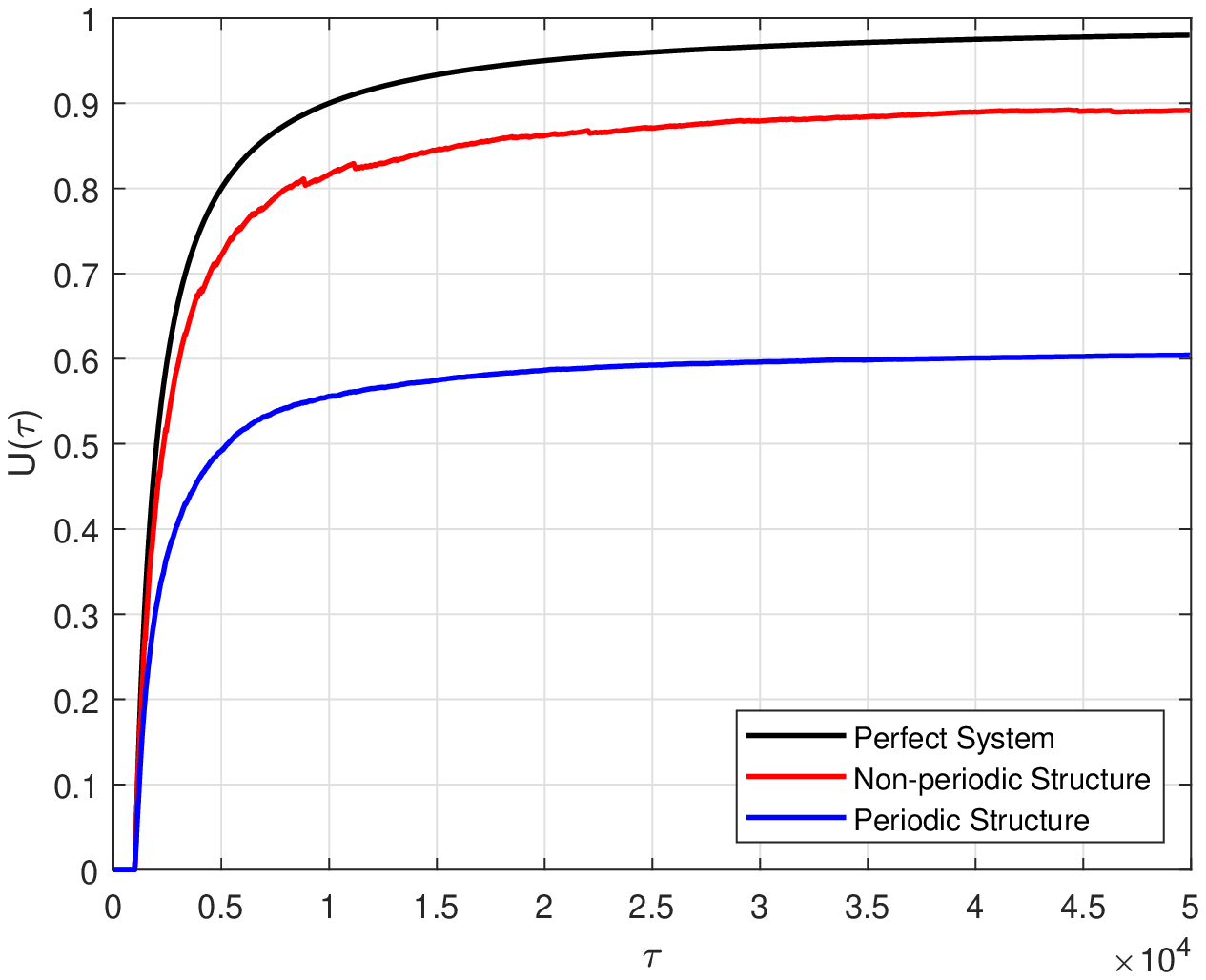}
			\label{Second_stage_1}
		}
		\subfigure[$\tau = 5 \times 10^4$.]{
			\includegraphics[width=0.45\textwidth]{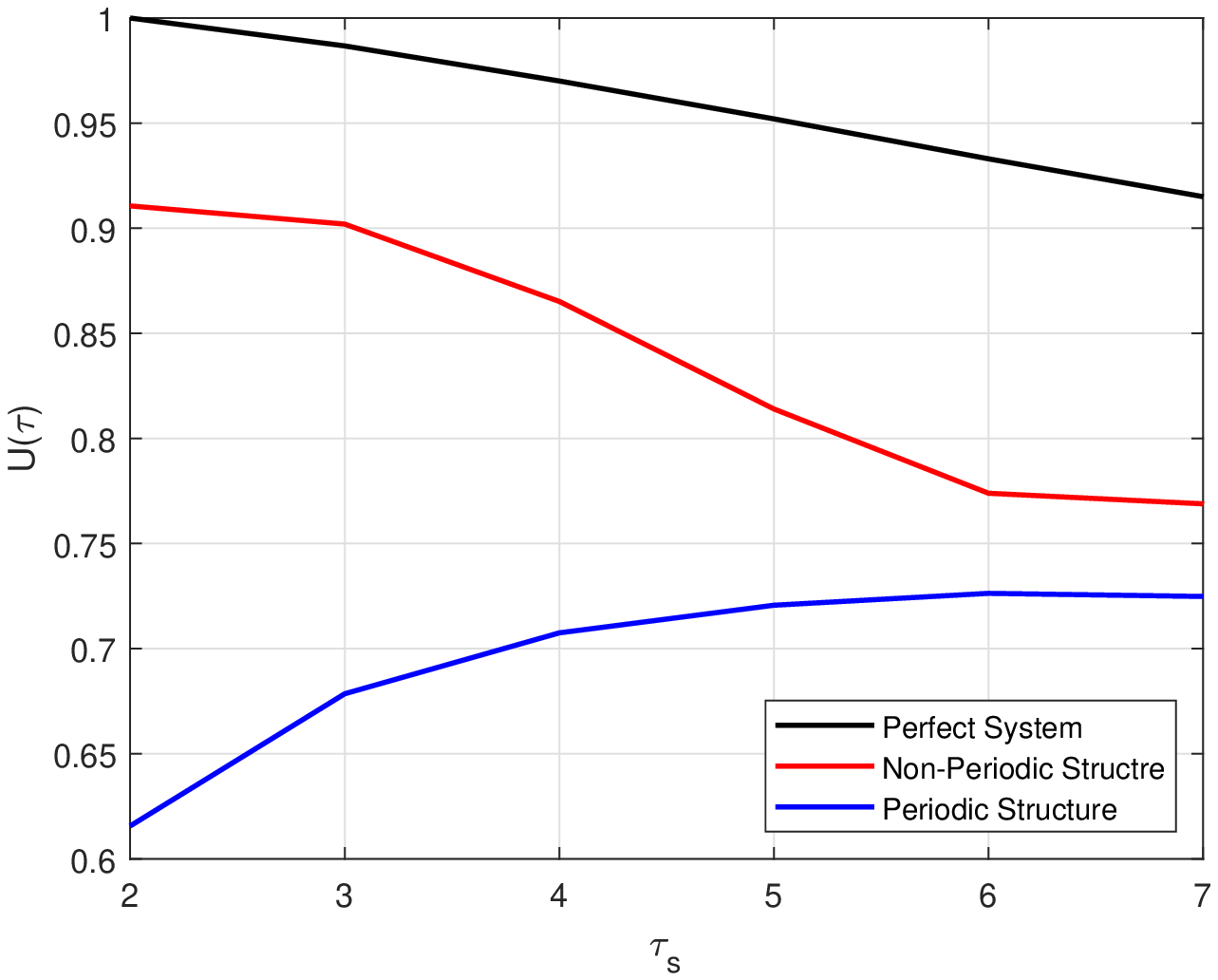}
			\label{Second_stage_2}
		}
		\caption{The NPLA performance $U(\tau)$ for the spectrum sharing strategies using different structures in Stage II with $f_s = 5$ MHz, $N^{(s)} = 10^4$, $\gamma_{st} = -12$ dB, and $\nu = 50$.}
		\label{Second_stage}
	\end{figure*}
	
	\subsection{Performance of the Two-Stage Spectrum Sharing Strategy with Online Function}
	\begin{figure}
		\centering
		\includegraphics[width=0.45\textwidth]{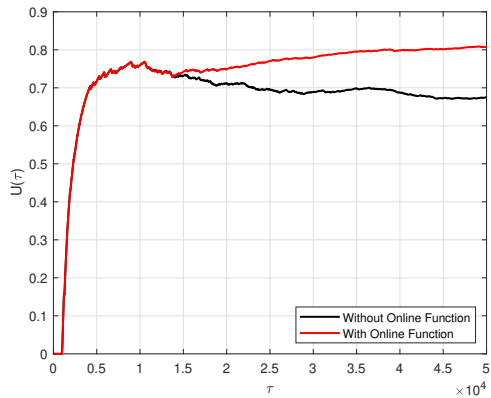}
		\caption{$U(\tau)$ for the spectrum sharing strategies with and without online function in an online scenario with $f_s = 5$ MHz, $N^{(s)} = 10^4$, $\tau_s = 4$, $\gamma_{st} = -12$ dB, and $\nu = 50$.}
		\label{Online Comparison}
	\end{figure}
	
	In Fig. \ref{Online Comparison}, we compare $U(\tau)$ of the proposed spectrum sharing strategies with and without the online function when the PT power level changes after Stage I. In this scenario, $L$ increases to 5 when $\tau = 10^4$, where $\gamma_{ps}^5 = 4 \gamma_{ps}^1$ and other $\gamma_{ps}^l$ remain the same. When $\tau \leqslant 10^4$, the online strategy has the same $U(\tau)$ performance with the regular strategy without the online function. When $\tau \geqslant 10^4$, $U(\tau)$ of the regular strategy decreases slowly. This is because the learning results becomes less accurate with an updated $L$, thus the existing strategy becomes less effective. In contrast, $U(\tau)$ of the online strategy first decreases when $10^4 \leqslant \tau \leqslant 1.4 \times 10^4$ due to the time consumption when updating $\textbf{X}$. When $\tau > 1.4 \times 10^4$, it is shown that $U(\tau)$ can continue to increase based on the updated knowledge of $L$.
	
	In summary, the above simulation results show that the learning method is beneficial to the design of spectrum sharing with multiple PT power levels. With the advent of 5G and industrial Internet of Things in the years to come, the proposed spectrum sharing provides a promising solution to the spectrum scarcity problem.

    \subsection{Discussion and Extension}
    \subsubsection{Rayleigh Fading Channel}
    For the convenience of the exposition of our proposed method, we have not considered fading. However, the extension to fading scenario is straightforward.  Specifically, for Rayleigh fading, $R_n$ in (\ref{R_i}) can be written as \cite{4489760}
    \begin{equation}
        R_{n}[i] = \sqrt{P_{PT,l}} h_n s_n[i]+ u_n[i], \quad \mathcal{H}_l, 
    \end{equation}
    where $h_n \sim {\cal CN}(0,1)$ is the fading channel coefficient. Accordingly, the test statistic $X_n$ still follows a Gaussian distribution as is the non-fading case, but with a different mean $(|h_n|^2 \gamma_{st}^l + 1) \sigma_u^2$ and variance $\frac{1}{N^{(s)}} \left( 2 |h_n|^2 \gamma_{st}^l + 1 \right)\sigma_u^4$. Clearly, when $h_n$ experiences deep fading, the mean $(|h_n|^2 \gamma_{st}^l + 1) \sigma_u^2$ will become very small and our solution is less capable of distinguishing the PT power levels. To solve this issue, motivated by \cite{4489760}, we can introduce time diversity by combining test statistics covering several consecutive time slots, i.e.,
    \begin{equation}
        Y_{n^*} = \sum _ { n_p = 1 } ^ { N_p } w_{n_p} X_{(n^* - 1)N_p + n_p },
    \end{equation}
    where $w_{n_p}$ is the weighting factor associated with the $n$-th sensing slot and $Y_{n^*}$ follows Gaussian distribution with mean $\sum_{n_p = 1}^{N_p} w_{n_p}(|h_{n_p}|^2 \gamma_{st}^l + 1) \sigma_n^2 $ and variance $\frac{1}{N^{(s)}} \sum_{n_p = 1}^{N_p} w_{n_p}^2 \left( 2 |h_{n_p}|^2 \gamma_{st}^l + 1 \right)\sigma_u^4 $.
    
    \subsubsection{Two Stages Switch}
    In the multi-armed bandit problem, dynamic switching between exploration and exploitation epochs \cite{6200864,6362216,8437583} achieves a regret with logarithmic order. However, the design of the two stages in our work are decoupled. That is, Stage I is a one-off design, which solely determines the accuracy in the model parameters estimation. Consequently, simply switching between Stage I and Stage II brings no benefits to the NPLA performance improvement, unless some information in Stage II (e.g., the correctness of the prediction) is fed to Stage I to improve the accuracy in the model parameters estimation. How to incorporate the information in Stage II into Stage I is an interesting problem to be considered in our future work.

    \subsubsection{Practical Implementation}
    Note that energy detection with classical signal processing could potentially suffer from hidden node problem, and can not perform well at very low SNRs. From machine learning perspective, it is possible to combine signal strength with other features to dramatically improve the performance. For example, much more accurate sensing performance at very low SNRs (e.g., -20~dB) can be achieved by integrating cyclic-prefix induced correlation as another feature with signal energy feature\cite{asss}. Therefore, the data-driven method may provide a potential solution to the hidden node problem. In Section II, we made the assumption that the test statistic $X_n$ follows a mixed Gaussian distribution. In a real scenario, such assumption might not hold as the real radio propagation environment is far more complicated. In this case, as suggested in \cite{8466022}, a kernel based process may be applied, which could be considered as our future work.

	\section{Conclusions}
	\label{Conclustion}
	In this paper, we considered the challenging problem of spectrum sharing where the PT transmits with multiple power levels. We endeavored to provide a practical solution with minimal or no prior information on the PT power profiles. We drew on the recent development in machine learning, and proposed a learning based two-stage spectrum sharing strategy. In Stage I, a conditionally conjugate DPGMM was proposed to capture the PT power level variation. Then, a Bayesian inference method was designed for model parameters inference to mathematically establish the PT power profile. Basically, Stage I depicts a big picture of the multi-level radio environments. Based on this knowledge, in Stage II, we designed prediction-transmission structures to enable the ST transmit power level to closely match that of the PT, minimizing the interference to the primary network. To achieve this, we relied on a new metric, NPLA, to characterize the extent of matching between PT and ST transmit power levels. To accommodate a more realistic scenario where the number of PT power levels might dynamically change after the learning, we developed an online function for the strategy. Finally, we verified the effectiveness of the proposed strategy with numerical results.
	
	The significance of our work for future impact are two-fold. On one hand, this data-driven based method may open a door to a new modality of spectrum sensing. On the other hand, our work may find its potential applications from industrial aspects. For example, the recent Ofcom\cite{ccccc} development on TV white space is considering both geolocation database and signal strength-based spectrum sensing.
	
	\appendix[Proof of Lemma \ref{lemma_2}]
	\label{Proof}
	We get the result by induction on $\tau$. In Lemma \ref{lemma_1}, we give the $V(\tau,p_{\tau}^k,k)$ being a convex function of $p_{\tau}^k$. Thus $V(\tau,p_{\tau}^k,k) \leqslant p_{\tau}^k V(\tau,1,k) = 0, \forall \tau \geqslant \mathcal{T}_k$, which means that $V(\tau,p_{\tau}^k,k)$ increases in $p_{\tau}^k$ for given $k$ with $\tau \geqslant \mathcal{T}_k$. Then, we prove the statement for $\tau < \mathcal{T}_k$.
	
	Let the statement holds for $V(\tau,p_{\tau}^k,k), \tau \geqslant \mathcal{T}_k-\upsilon$, then Lemma \ref{lemma_2} can be proved if both terms in \eqref{optimality equation} increase in $p_{\tau}^k$. Thus, we give the first derivative of $E(\tau,p_{\tau}^k,k)$, which is expressed as
	\begin{equation}
	\label{deri_L}
	\dfrac{d E(\tau,p_{\tau}^k,k) }{dp_{\tau}^k} = \sum_{j = 1}^{L} g_{\tau}^S (H_{jk} - \boldsymbol{c}_{k} \boldsymbol{h}_{j}) V(\tau+1,p_{\tau + 1}^k(\mathcal{H}_j),k) + \sum_{j = 1}^{L} \dfrac{g_{\tau}^S H_{jk} \boldsymbol{c}_{k} \boldsymbol{h}_{j}  V'(\tau+1,p_{\tau + 1}^k(\mathcal{H}_j),k)}{p_{\tau}^k g_{\tau}^S H_{jk} + (1 - p_{\tau}^k g_{\tau}) \boldsymbol{c}_{k} \boldsymbol{h}_{j}}.
	\end{equation}
	As the discretized sensing time equals to 1, we have $V'(\tau+1,p_{\tau + 1}^k (\mathcal{H}_j),k) > 0$ for $\tau \geqslant \mathcal{T}_k - \upsilon -1$, thus the second term in \eqref{deri_L} holds positive. In order to prove the first term in \eqref{deri_L} positive, we define	
	\begin{equation}
	\label{Cap_F_1}
	\begin{split}
	\mathcal{C} & = \sum_{j = 1}^{L} (H_{jk} - \boldsymbol{c}_{k} \boldsymbol{h}_{j}) V(\tau+1,p_{\tau +1}^k(\mathcal{H}_j),k) \\
	& = (H_{kk} - \boldsymbol{c}_{k} \boldsymbol{h}_{k}) V(\tau+1,p_{\tau + 1 }^k(\mathcal{H}_k),k) + \sum_{j = 1, j \neq k}^{L} (H_{jk} - \boldsymbol{c}_{k} \boldsymbol{h}_{j}) V(\tau + 1,p_{\tau + 1}^k(\mathcal{H}_j),k).
	\end{split}
	\end{equation}
	In \eqref{Cap_F_1}, we define
	\begin{equation}
	\mathcal{D}(k,j) = H_{jk} - \boldsymbol{c}_{k} \boldsymbol{h}_{j} < H_{jk} - H_{jj} C_{jk}.
	\end{equation}
	In this paper, we consider such a case that the PT transfers to different states with similar probabilities ($C_{jk} \approx 1/(L - 1), j \neq k$). Besides, the probability of the ST correctly clustering the signals is much larger than incorrectly clustering ($H_{jj} \gg L H_{jk}, j \neq k$), which can be verified by the numerical results. Thus $\mathcal{D}(k,j) < 0$. Besides, we define
	\begin{equation}
	\begin{split}
	& \mathcal{F} = p_{\tau + 1}^k(\mathcal{H}_k) - p_{\tau + 1}^k(\mathcal{H}_j) \\
	& = \dfrac{p_{\tau}^k g_{\tau}^S (1 - p_{\tau}^k g_{\tau}^S) \left( H_{kk}\boldsymbol{c}_{k} \boldsymbol{h}_{j} - H_{jk}\boldsymbol{c}_{k} \boldsymbol{h}_{k} \right) }{\left[ p_{\tau}^k g_{\tau}^S H_{kk} + (1 - p_{\tau}^k g_{\tau}^S)\boldsymbol{c}_{k} \boldsymbol{h}_{k} \right] \left[ p_{\tau}^k g_{\tau}^S H_{jk} + (1 - p_{\tau}^k g_{\tau}^S) \boldsymbol{c}_{k} \boldsymbol{h}_{j}\right] } \\
	& = \dfrac{p_{\tau}^k g_{\tau}^S (1 - p_{\tau}^k g_{\tau}^S) \sum_{i = 1}^{L} \left( H_{kk} H_{j i} - H_{j k} H_{k i} \right) C_{i k} }{\left[ p_{\tau}^k g_{\tau}^S H_{kk} + (1 - p_{\tau}^k g_{\tau}^S)\boldsymbol{c}_{k} \boldsymbol{h}_{j} \right] \left[ p_{\tau}^k g_{\tau}^S H_{jk} + (1 - p_{\tau}^k g_{\tau}^S) \boldsymbol{c}_{k} \boldsymbol{h}_{j}\right] }.
	\end{split}
	\end{equation}
	Similarly, we can get $\mathcal{F} > 0$ for $j \neq k$. Thus, $V(\tau+1,p_{\tau+1}^k(\mathcal{H}_k),k) > V(\tau + 1,p_{\tau + 1}^k(\mathcal{H}_j),j)$ for $j \neq k$ and $\tau \geqslant \mathcal{T}_k-\upsilon$.
	
	Therefore, we can reach the following result.
	\begin{equation}
	\begin{split}
	\mathcal{C} & > (H_{kk} - \boldsymbol{c}_{k} \boldsymbol{h}_{k}) V(\tau+1,p_{\tau+1}^k(\mathcal{H}_k),k) + \sum_{j = 1, j \neq k}^{L} (H_{jk} - \boldsymbol{c}_{k} \boldsymbol{h}_{j}) V(\tau + 1,p_{\tau + 1}^k(\mathcal{H}_k),k) \\
	& = \left( \sum_{j=1}^{L} H_{jk} - \sum_{j = 1}^{L} \sum_{i = 1}^{L} H_{ji} C_{ik} \right) V(\tau+1,p_{\tau+1}^k(\mathcal{H}_k),k) \\
	& = \left(1 - \sum_{i = 1}^{L} C_{ik} \right) V(\tau+1,p_{\tau+1}^k(\mathcal{H}_k),k) \\
	& = 0.
	\end{split}
	\end{equation}
	Therefore, the first term in \eqref{deri_L} is proved positive. Thus, the $L(\tau,p_{\tau}^k,k)$ increases in $p_{\tau}^k$ for $\tau = \mathcal{T}_k - \upsilon - 1$. Similarly, the first derivative of $A(\tau,p_{\tau}^k,k)$ can be given by
	\begin{equation}
	\label{deriv_M}
	\begin{split}
	\dfrac{d A(\tau,p_{\tau}^k,k)}{dp_{\tau}^k} = & g_{\tau}^T (1 - \sum_{j=1}^{L} C_{kj}) V(\tau+\tau_s, p_{\tau+\tau_s}^k(\mathcal{A}),k) \\
	& + \dfrac{g_{\tau}^T \sum_{j=1}^{L} C_{kj} V'(\tau+\tau_s, p_{\tau+\tau_s}^k(\mathcal{A}),k)}{p_{\tau}^k g_{\tau}^T + (1 - p_{\tau}^k g_{\tau}^T) \sum_{j = 1}^{L} C_{kj}}  + \left(  D_k + \sum_{j = 1}^{L} C_{kj} Y_j \right)  g_{\tau}^T \tau_s.
	\end{split}
	\end{equation}
	According to \eqref{C_kf} and \eqref{T*}, we can get $\sum_{j = 1}^{L} C_{kj} = 1$ and $D_k + \sum_{j = 1}^{L} C_{kj} Y_j > 0$, which means \eqref{deriv_M} is positive and $A(\tau,p_{\tau}^k,k)$ increases in $p_{\tau}^k$ for $\tau = \mathcal{T}_k - \upsilon - 1$. In summary, both the terms in the \eqref{optimality equation} increase in $p_{\tau}^k$ for $\tau = \mathcal{T}_k - \upsilon - 1$ and we complete the proof.
	
	\bibliography{document}
	\bibliographystyle{IEEEtran}
	
\end{document}